\documentclass[twoside]{aiml}
\usepackage{aimlmacro}

\usepackage{graphicx}
\usepackage{amsmath}
\usepackage{amssymb}

\usepackage{tikz}
\usepackage{xspace}
\usepackage{color}

\tikzset{>=latex, 
	point/.style = {circle,draw,thick,minimum size=2mm,inner sep=0pt},
	point1/.style = {circle,draw,thick,minimum size=6mm,inner sep=0pt},
	hm/.style = {dotted,semithick},
	role/.style = {thick},
	tree/.style = {rounded corners=10pt, dashed, fill opacity=0.5, fill=nullscolour},
	wiggly/.style={thick,
	},
	query/.style={thick},
	itria/.style={
  draw,dashed,shape border uses incircle,
  isosceles triangle,shape border rotate=90,yshift=-1.45cm},
  square/.style={regular polygon,regular polygon sides=4}
}

\newcommand{\sig}{\textit{sig}}
\newcommand{\sub}{\textit{sub}}

\newcommand{\K}{\mathsf{K}}
\newcommand{\KF}{\mathsf{K4}}
\newcommand{\wKF}{\mathsf{wK4}}
\newcommand{\SF}{\mathsf{S4}}
\newcommand{\SFI}{\mathsf{S5}}
\newcommand{\DL}{\mathsf{DL}}

\newcommand{\bis}{\boldsymbol{\beta}}

\newcommand{\type}{t}
\newcommand{\p}{\boldsymbol{w}}

\newcommand{\Type}{\tau}

\newcommand{\at}{\textit{at}}

\newcommand{\Rs}{R^s}
\newcommand{\Rr}{R^r}

\newcommand{\opluss}{\!\oplus\!}
\newcommand{\avar}{\mathsf{a}}
\newcommand{\bvar}{\mathsf{b}}
\newcommand{\evar}{\mathsf{e}}
\newcommand{\lvar}{\mathsf{level}}

\newcommand{\tvar}{\mathsf{t}}
\newcommand{\grid}{\mathsf{grid}}
\newcommand{\gridx}{\mathsf{succ}_x}
\newcommand{\gridy}{\mathsf{succ}_y}

\newcommand{\atype}{a}
\newcommand{\btype}{b}
\newcommand{\simr}{\sim^\varrho}
\newcommand{\equC}{\approx}
\newcommand{\cC}[1]{[{#1}]}
\newcommand{\TC}[1]{T^\varrho[{#1}]}
\newcommand{\ATC}[1]{AT^\varrho[{#1}]}
\newcommand{\ftype}{\boldsymbol{t}}
\newcommand{\MM}[1]{M_{[{#1}]}}
\newcommand{\iset}[1]{I_{#1}}

\newcommand{\rr}{r}
\newcommand{\dom}{\ensuremath{{\sf dom}}\xspace}

\newcommand{\PSpace}{\textnormal{\sc PSpace}\xspace}
\newcommand{\coNExpTime}{\textnormal{\sc coNExpTime}\xspace}
\newcommand{\NExpTime}{\textnormal{\sc NExpTime}\xspace}

\def\lastname{Kurucz, Wolter and Zakharyaschev}

\begin{document}

\begin{frontmatter}
  \title{The Interpolant Existence Problem for Weak K4 and Difference Logic}
   \author{Agi Kurucz}
  \address{Department of Informatics\\
  King's College London, U.K.}
  \author{Frank Wolter}
  \address{Department of Computer Science\\
  University of Liverpool, U.K.}
    \author{Michael Zakharyaschev}
  \address{School of Computing and Mathematical Sciences\\
  Birkbeck, University of London, U.K.}
\begin{abstract}
As well known, weak $\KF$ and the difference logic $\DL$ do not enjoy the Craig interpolation property. Our concern here is the problem of deciding whether any given implication does have an interpolant in these logics. We show that the nonexistence of an interpolant can always be witnessed by a pair of bisimilar models of polynomial size for $\DL$ and of triple-exponential size for weak $\KF$, and so the interpolant existence problems for these logics are decidable in \textsc{coNP} and \textsc{coN3ExpTime}, respectively. We also establish \textsc{coNExpTime}-hardness of this problem for weak $\KF$, which is higher than the \textsc{PSpace}-completeness of its decision problem. 
\end{abstract}

 \end{frontmatter}


\section{Introduction}

Weak $\KF$ is the modal logic one obtains when the $\Diamond$-operator of the propositional classical (uni)modal language is interpreted by the derivative operation\footnote{The \emph{derived set} of a subset $X$ of a topological space comprises all limit points of $X$, i.e., those $x$ all of whose neighbourhoods contain a point in $X$ different from $x$.} in topological spaces~\cite{Esakia2001}---rather than the more conventional topological closure, which results in classical $\SF$~\cite{McKinsey&Tarski44}. 
In terms of Kripke semantics, weak $\KF$ is characterised~\cite{Esakia2001} by the class of weakly transitive frames, i.e., those $\mathfrak F = (W,R)$ that satisfy the condition 
\begin{equation}\label{weaktr}
\forall x,y,z \in W\, \big( xRyRz \to (x=z) \vee xRz \big),
\end{equation}
which explains the moniker `weak $\KF$' or $\wKF$ for this logic. 
Syntactically, $\wKF$ is obtained by adding the axiom $\Diamond\Diamond p \rightarrow (p \vee \Diamond p)$ to the basic normal modal logic $\K$.
A notable extension of $\wKF$ is the difference logic $\DL$, which goes back to the `logic of elsewhere' \cite{vonWright79,Segerberg1980-SEGANO-2} and can be axiomatised by adding  
the Brouwersche axiom $p \to \Box \Diamond p$ to $\wKF$. While an arbitrary frame for $\DL$ is a symmetric and weakly transitive relation, $\DL$ is also known to be characterised by the class of \emph{difference frames} (that is, Kripke frames of the form $(W, \ne)$)~\cite{Segerberg1980-SEGANO-2}. 


Despite their apparent similarity to $\KF$, $\SF$ and $\SFI$, the logics $\wKF$ and $\DL$ have---or rather lack---one important feature: they do not enjoy the Craig interpolation property (CIP)~\cite{Karpenko&Maksimova2010}, according to which each valid implication $\varphi \to \psi$ in a logic $L$ has an interpolant in $L$, viz.\ a formula $\iota$ built from common variables of $\varphi$ and $\psi$ such that $(\varphi \to \iota) \in L$ and $(\iota \to \psi) \in L$; if $\varphi$ and $\psi$ have no variables in common, variable-free interpolant $\iota$ is built from the logical constants $\top$ and $\bot$. 


\begin{example}\label{ex:basic}
In the pictures below, $\bullet$ always denotes an irreflexive point, $\circ$ a reflexive one, and an ellipse represents a cluster
(a set of points, in which any two distinct ones `see' each other). 

$(i)$ Consider the following formulas without common variables:
\[
\varphi = \Diamond\Diamond p \land \neg \Diamond p, \qquad  \psi = \Diamond \Diamond \neg q \lor q \ (\equiv \Box\Box q\rightarrow q).
\]
It is easy to see that $\varphi \to \psi$ is true in all models based on weakly transitive frames, and so $(\varphi \to \psi) \in \wKF$. On the other hand, the picture below shows models $\mathfrak M_\varphi$ and $\mathfrak M_\psi$ based on weakly transitive frames such that $\mathfrak M_\varphi,\rr_\varphi \models \varphi$, $\mathfrak M_\psi,\rr_\psi \models \neg \psi$, and the 
universal relation $\bis$ between the points of $\mathfrak M_\varphi$ and $\mathfrak M_\psi$ is a $\varrho$-bisimulation for the shared signature $\varrho = \emptyset$ of $\varphi$ and $\psi$, with $\rr_\varphi \bis \rr_\psi$ (see Sec.~\ref{prelims} for definitions).\\
%
%
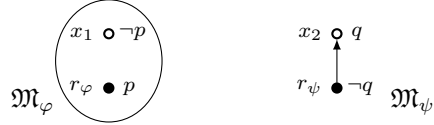
\centerline{\begin{tikzpicture}[>=latex,line width=0.4pt,xscale = 1,yscale = .9]
\node at (-1,-.6) {$\mathfrak M_\varphi$};
\draw[] (0,0) ellipse (.7 and .9);
\node[point,scale = 0.6,label=left:{\footnotesize $x_1$},label=right:\!{\footnotesize $\neg p$}] (f1) at (0,.4) {};
\node[point,fill=black,scale = 0.6,label=left:{\footnotesize $\rr_\varphi$},label=right:{\footnotesize $p$}] (f0) at (0,-.4) {};
\node at (4,-.6) {$\mathfrak M_\psi$};
\node[point,scale = 0.6,label=left:{\footnotesize $x_2$},label=right:{\footnotesize $q$}] (p1) at (3,.4) {};
\node[point,fill=black,scale = 0.6,label=left:{\footnotesize $\rr_\psi$},label=right:\!{\footnotesize $\neg q$}] (p0) at (3,-.4) {};
\draw[->] (p0) to (p1);
\end{tikzpicture}
}
%
It follows that there is no variable-free formula $\iota$ with $(\varphi \to \iota) \in \wKF$ and $(\iota \to \psi) \in \wKF$ 
because $\varrho$-bisimulations preserve the truth-values of $\varrho$-formulas.

$(ii)$ Consider the formulas 
\[
\varphi =  \Diamond\Diamond p \land \neg \Diamond p\land\neg\Diamond\Diamond\Diamond p, \qquad  
 \psi = \Diamond\Diamond s \land \neg \Diamond s\to\neg\bigl(\Diamond(q\land\Diamond s)\land\Diamond(\neg q\land\Diamond s)\bigr).
\]
It is easy to see that $\varphi \to \psi$ is true in all models based on a weakly transitive frame, and so $(\varphi \to \psi) \in \wKF$.\\
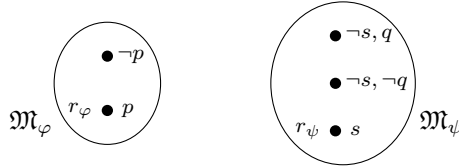
\centerline{\begin{tikzpicture}[>=latex,line width=0.4pt,xscale = 1,yscale = .9]
\node at (-1,-.6) {$\mathfrak M_\varphi$};
\draw[] (0,0) ellipse (.7 and .9);
\node[point,fill=black,scale = 0.6,label=right:\!{\footnotesize $\neg p$}] (f1) at (0,.4) {};
\node[point,fill=black,scale = 0.6,label=left:{\footnotesize $\rr_\varphi$},label=right:{\footnotesize $p$}] (f0) at (0,-.4) {};
\node at (4.4,-.6) {$\mathfrak M_\psi$};
\draw[] (3.1,0) ellipse (.95 and 1.2);
\node[point,fill=black,scale = 0.6,label=right:\!{\footnotesize $\neg s,q$}] (p2) at (3,.7) {};
\node[point,fill=black,scale = 0.6,label=right:\!{\footnotesize $\neg s,\neg q$}] (p1) at (3,0) {};
\node[point,fill=black,scale = 0.6,label=left:{\footnotesize $\rr_\psi$},label=right:{\footnotesize $s$}] (p0) at (3,-.7) {};
\end{tikzpicture}
}
%
On the other hand, the models $\mathfrak M_\varphi$ and $\mathfrak M_\psi$ above are 
based on difference frames with $\mathfrak M_\varphi,\rr_\varphi \models \varphi$, $\mathfrak M_\psi,\rr_\psi \models \neg \psi$, and 
the universal relation $\bis$ between the points of $\mathfrak M_\varphi$ and $\mathfrak M_\psi$ is a $\varrho$-bisimulation
with $\rr_\varphi \bis \rr_\psi$, for the shared signature $\varrho = \emptyset$ of $\varphi$, $\psi$.
Therefore, $\varphi$ and $\psi$ do not have an interpolant in any logic between $\wKF$ and $\DL$. \hfill $\dashv$
\end{example}

Our concern in this paper is the following 
\begin{description}
\item[\emph{interpolant existence problem} (\emph{IEP}) for $L \in \{\wKF, \DL\}$:]  given formulas $\varphi$ and $\psi$, decide whether $\varphi \to \psi$ has an interpolant in $L$. 
\end{description}
We show that the IEP for $\wKF$ is decidable in \textsc{coN3ExpTime}, being \textsc{coNExpTime}-hard (harder than its decision problem), while the IEP for $\DL$ is \textsc{coNP}-complete (as is its decision problem). 

In Section~\ref{prelims}, we introduce the necessary technical tools and demonstrate them in the case of $\DL$. Then we focus on the much more involved IEP for $\wKF$, establishing the upper bound in Section~\ref{sec:deciding} and the lower one in Section~\ref{lower}.
Finally, we discuss related and open problems in  Section~\ref{discussion}.
 

\section{Preliminaries}\label{prelims}

All (Kripke) frames $\mathfrak F = (W,R)$ we deal with in this paper are assumed to be \emph{weakly transitive}~\eqref{weaktr}. 
By a \emph{cluster} in $\mathfrak F$ we mean any set of the form 
$$
C(x) = \{x\} \cup \{ y \in W \mid xRy \land y R x\}, \quad x\in W. 
$$
A cluster $\mathfrak F$ can contain \emph{reflexive} points $y$, for which $yRy$, as well as \emph{irreflexive} points $z$, for which $\neg (zRz)$. A cluster with a single point, which is irreflexive, is said to be \emph{degenerate}. Given $x,y \in W$, we write $xR^sy$ iff $xRy$ and $C(x) \ne C(y)$.

Suppose $C$ and $C'$ are clusters in $\mathfrak F$ and $x \in W$. We write $CRx$ if there exists $y\in C$ such that $yRx$, and $CRC'$ if there are $x\in C$ and $y\in C'$ with $xRy$. Observe that $CRC$ iff $C$ is non-degenerate. 
We write $C\Rs x$ if $CRx$ and $x\not\in C$, and $C\Rs C'$ if $CRC'$ and $C \ne C'$. Thus, $\Rs$ is a strict partial order on the set $W_c$ of clusters in $\mathfrak F$. The reflexive closure of $\Rs$ is denoted by $\Rr$. The frame $\mathfrak{F}$ is called \emph{rooted} if there is $r \in W$, a \emph{root} of $\mathfrak F$, with $W_c = \{C(x) \mid C(r) R^r C(x)\}$. 

By a \emph{signature} we mean any \emph{finite} set of propositional variables, $p_i$. Given a signature $\sigma$, a $\sigma$-\emph{formula} is built from variables in $\sigma$ and logical constants $\bot$, $\top$ using the Boolean connectives $\land$, $\neg$ and the modal possibility operator $\Diamond$. The other Boolean connectives and the necessity operator $\Box$ are regarded as standard abbreviations. We denote by $\sig(\varphi)$ the set of variables in a formula $\varphi$ and by $\sub(\varphi)$ the set of subformulas of $\varphi$ together with their negations,
setting $|\varphi|=|\sub(\varphi)|$. 
We also use abbreviations $\Diamond^+\varphi= \varphi\lor\Diamond\varphi$, $\Box^+\varphi=\varphi\land\Box\varphi$ and $\Diamond\Gamma=\{\Diamond\varphi\mid\varphi\in\Gamma\}$, for a set $\Gamma$ of formulas.
 

A $\sigma$-\emph{model} based on a frame $\mathfrak{F}=(W,R)$ is a pair $\mathfrak{M}=(\mathfrak{F},\mathfrak{v})$ with a \emph{valuation} $\mathfrak{v} \colon \sigma \to 2^W$. 
%
%
The \emph{truth-relation} $\mathfrak M,x \models \varphi$, for $x \in W$ and a $\sigma$-formula $\varphi$, is defined by induction as usual in Kripke semantics (e.g., $\mathfrak M,x \models \Diamond \varphi$ iff $\mathfrak M,y \models \varphi$, for some $y \in W$ with $xRy$). 
%
For any $\varrho \subseteq \sigma$, the \emph{$\varrho$-type of $x \in W$ in} $\mathfrak{M}$ is the set $t^\varrho_{\mathfrak{M}}(x)$ of all $\varrho$-formulas that are true at $x$ in $\mathfrak M$,
and the \emph{atomic $\varrho$-type of $x \in W$ in} $\mathfrak{M}$ is $\at^\varrho_{\mathfrak{M}}(x)=t^\varrho_{\mathfrak{M}}(x)\cap\varrho$. 
 For a set $X$ of points in $\mathfrak M$, we let
$t_{\mathfrak{M}}^{\varrho}(X)=\bigl\{t_{\mathfrak{M}}^{\varrho}(x)\mid x\in X\bigr\}$
and $\at_{\mathfrak{M}}^{\varrho}(X)=\bigl\{\at_{\mathfrak{M}}^{\varrho}(x)\mid x\in X\bigr\}$. 
A set $\Gamma$ of $\sigma$-formulas is \emph{finitely satisfiable} in $\mathfrak M$ if, for every finite $\Gamma' \subseteq \Gamma$, there is $x' \in W$ such that $\Gamma' \subseteq t^\sigma_{\mathfrak{M}}(x')$; $\Gamma$ is \emph{satisfiable in} $\mathfrak M$ if $\Gamma \subseteq t^\sigma_{\mathfrak{M}}(x)$, for some $x \in W$. 

A $\sigma$-model $\mathfrak M$ is \emph{descriptive} if, for any $x,y \in W$ and any set $\Gamma$ of $\sigma$-formulas,
\begin{description}
\item[(dif)] $x = y$ iff $t^\sigma_{\mathfrak{M}}(x) = t^\sigma_{\mathfrak{M}}(y)$,

\item[(ref)] $xRy$ iff $\Diamond t^\sigma_{\mathfrak{M}}(y) \subseteq t^\sigma_{\mathfrak{M}}(x)$ iff 
$\{\varphi\mid\Box\varphi\in t^\sigma_{\mathfrak{M}}(x)\}\subseteq t^\sigma_{\mathfrak{M}}(y)$,

\item[(com)] if $\Gamma$ is finitely satisfiable in $\mathfrak M$, then $\Gamma$ is satisfiable in $\mathfrak M$. 
\end{description}
(In other words, descriptive models are based on finitely generated  descriptive frames for $\wKF$~\cite{DBLP:books/daglib/0030819}.) 
%
We remind the reader that, for any $\sigma$-formula $\varphi$, we have $\varphi \in \wKF$ iff $\neg \varphi$ is not satisfiable in a (finite) $\sigma$-model iff $\neg \varphi$ is not satisfiable in a (finite) descriptive $\sigma$-model. The finite model property of $\wKF$ was established in~\cite{DBLP:conf/tbillc/BezhanishviliEG09,DBLP:journals/ndjfl/BezhanishviliGJ11,kudinov2024filtrations}. The decision problem for $\wKF$ is $\PSpace$-complete~\cite{DBLP:journals/tocl/Shapirovsky22}. 
For the \emph{difference logic} $\DL$, we have $\varphi \in \DL$ iff $\neg \varphi$ is not satisfiable in a
polynomial-size $\sigma$-model based on a frame for $\DL$; 
 the decision problem for $\DL$ is \textsc{coNP}-complete~\cite{DBLP:journals/jsyml/Rijke92}.

In the remainder of this section, we present the technical tools and results we need for deciding the interpolant existence problem for $\DL$ and $\wKF$.
Given $\varrho \subseteq \sigma$, we call a cluster $C$ \emph{$\varrho$-maximal} in $\mathfrak{M}$ if, for any $x \in C$ and $y \in W$, whenever $CR y$ and $t^\varrho_{\mathfrak{M}}(x) = t^\varrho_{\mathfrak{M}}(y)$, then $y \in C$. 
The following fundamental properties of descriptive $\sigma$-models for $\wKF$ are similar to the corresponding well-known properties of finitely generated descriptive frames for $\KF$~\cite{DBLP:journals/jsyml/Fine74,DBLP:books/daglib/0030819}:


\begin{lemma}\label{l:fundament}
Suppose $\mathfrak M$ is a descriptive $\sigma$-model based on some $\wKF$-frame $(W,R)$,
$\varrho \subseteq \sigma$, $C$ is a cluster in $\mathfrak M$, and $\Gamma$ a set of $\sigma$-formulas. Then the following hold\textup{:}
\begin{description}
\item[$(a)$] $|C| \le 2^{|\sigma|}$\textup{;}

\item[$(b)$] if $\mathfrak M, x\models \Diamond \bigwedge\Gamma'$ for every finite $\Gamma' \subseteq \Gamma$, then there is $y$ such that $xRy$ and $\mathfrak M, y \models \Gamma$\textup{;} 

\item[$(c)$]
if $\Diamond t^\varrho_{\mathfrak M}(z)\subseteq t^\varrho_{\mathfrak M}(x)$, then there is $y$ such that $xRy$ and 
$t^\varrho_{\mathfrak M}(y)= t^\varrho_{\mathfrak M}(z)$\textup{;}

\item[$(d)$] there exists a $\varrho$-maximal cluster $C'$ such that $C\Rr C'$ and $t^\varrho_{\mathfrak M}(C)\subseteq t^\varrho_{\mathfrak M}(C')$.
\end{description}
\end{lemma}
\begin{proof}
$(a)$ It is easy to see that if $x,y \in C$ and $\at^\sigma_{\mathfrak{M}}(x) = \at^\sigma_{\mathfrak{M}}(y)$, then $t^\sigma_{\mathfrak{M}}(x) = t^\sigma_{\mathfrak{M}}(y)$.
It follows by {\bf (dif)} that $|C| \le 2^{|\sigma|}$.

$(b)$ By the assumption, the set $\Xi = \Gamma \cup \{ \varphi \mid \Box\varphi \in t^\sigma_{\mathfrak M}(x) \}$ is finitely satisfiable, and so, by {\bf (com)}, $\mathfrak{M},y\models \Xi$, for some $y$, with $xRy$ by {\bf (ref)}.

$(c)$ is a straightforward consequence of $(b)$ with $\Gamma=t^\varrho_{\mathfrak M}(z)$.

$(d)$ Take any $x \in C$ and consider the set
\begin{multline*}
\mathcal{X}=\{X'\subseteq W\mid \mbox{$x\in X'$, $t^\varrho_{\mathfrak M}(x') = t^\varrho_{\mathfrak M}(x)$ for all $x' \in X'$, and}\\
R^s\cap(X'\times X')\ \mbox{is a strict linear order with smallest element $x$}\}.
\end{multline*}
By Zorn's lemma, there is a $\subseteq$-maximal set $X^\dag$ in $\mathcal{X}$. 
We claim that there is an $R^s$-maximal point in $X^\dag$. Indeed, suppose otherwise.
Then the set
\[
\Xi = t^\varrho_{\mathfrak M}(x) \cup \{\varphi \mid \Box\varphi \in t^\sigma_{\mathfrak M}(y) \mbox{ for some $y\in X^\dag$}\}
\]
is finitely satisfiable. By {\bf (com)}, there is $z$ satisfying $\Xi$, so $t^\varrho_{\mathfrak{M}}(z) =  t^\varrho_{\mathfrak{M}}(x)$ and, by {\bf (ref)}, $yRz$ for all $y\in X^\dag$. Two cases are possible now. $(i)$ If $yR^sz$ for all $y\in X^\dag$, then $X^\dag\cup\{z\}\in\mathcal{X}$, and so 
$z\in X^\dag$ by the $\subseteq$-maximality of $X^\dag$ in $\mathcal{X}$. Thus, $z$ is an $R^s$-maximal point in 
$X^\dag$, contrary to our assumption. $(ii)$ If $zRy$, for some $y\in X^\dag$, then by our assumption there is $y'\in X^\dag$ with $yR^s y'$, and so $y'Rz$. By \eqref{weaktr}, either $y'=y$ or $y'R y$, contrary to $yR^s y'$.

Now, let $x^\dag$ be an $R^s$-maximal point in $X^\dag$.
We claim that $C(x^\dag)$ is $\varrho$-maximal. Suppose otherwise and there exists $z \in C(x^\dag)$ with an $R$-successor $z'\notin C(x^\dag)$ such that $t^\varrho_{\mathfrak M}(z) = t^\varrho_{\mathfrak M}(z')$. Then $x^\dag R^s z'$.
Since $z \in C(x^\dag)$, either $(i)$ $z=x^\dag$, and so  
$t^\varrho_{\mathfrak M}(z') = t^\varrho_{\mathfrak M}(x^\dag) =t^\varrho_{\mathfrak M}(x)$,
or $(ii)$ $zRx^\dag$, and so 
$\Diamond t^\varrho_{\mathfrak M}(x)=\Diamond t^\varrho_{\mathfrak M}(x^\dag)\subseteq t^\varrho_{\mathfrak M}(z)=
t^\varrho_{\mathfrak M}(z')$. By $(c)$, this implies  that there is an $R$-successor $z''$ of $z'$ with $t^\varrho_{\mathfrak M}(z'') = t^\varrho_{\mathfrak M}(x)$.
So in both $(i)$ and $(ii)$, there is an
$R^s$-successor $z''$ of $x^\dag$ with $t^\varrho_{\mathfrak M}(z'') = t^\varrho_{\mathfrak M}(x)$.
By the $R^s$-maximality of $x^\dag$ in $X^\dag$, it follows that $X^\dag\cup\{z''\}\in\mathcal{X}$ and
$z''\notin X^\dag$, contrary to the $\subseteq$-maximality of $X^\dag$ in $\mathcal{X}$.

A similar argument shows that $t^\varrho_{\mathfrak M}(C(x))\subseteq t^\varrho_{\mathfrak M}\bigl(C(x^\dag)\bigr)$. Let $z\in C(x)$. If $z=x$, then clearly  
$t^\varrho_{\mathfrak M}(z)\in t^\varrho_{\mathfrak M}\bigl(C(x^\dag)\bigr)$. So suppose $z\ne x$ and hence 
$\Diamond t^\varrho_{\mathfrak M}(z)\subseteq t^\varrho_{\mathfrak M}(x)=t^\varrho_{\mathfrak M}(x^\dag)$.
By $(c)$, there is an $R$-successor $z'$ of $x^\dag$ with $t^\varrho_{\mathfrak M}(z') = t^\varrho_{\mathfrak M}(z)$. We claim that $z'\in C(x^\dag)$. Indeed, otherwise we have $x^\dag R^s z'$. As  $\Diamond t^\varrho_{\mathfrak M}(x)\subseteq t^\varrho_{\mathfrak M}(z)=t^\varrho_{\mathfrak M}(z')$, by $(c)$ there is an $R$-successor $z''$ of $z'$ with $t^\varrho_{\mathfrak M}(z'') = t^\varrho_{\mathfrak M}(x)$.
By \eqref{weaktr}, it follows that $x^\dag R^s z''$. 
Thus, by the $R^s$-maximality of $x^\dag$ in $X^\dag$, we have $X^\dag\cup\{z''\}\in\mathcal{X}$ and
$z''\notin X^\dag$, contrary to the $\subseteq$-maximality of $X^\dag$ in $\mathcal{X}$.

As $C(x)R^r C(x^\dag)$, cluster $C'=C(x^\dag)$ is as required.
%
%
\end{proof}

Let $\mathfrak{M}_i$, $i=1,2$, be $\sigma$-models based on frames $\mathfrak F_i = (W_i,R_i)$ for $\wKF$ and let $\varrho \subseteq \sigma$. 
A relation $\bis \subseteq W_1 \times W_2$ is called a $\varrho$-\emph{bisimulation} between $\mathfrak M_1$ and $\mathfrak M_2$ in case  the following conditions hold: whenever $x_1 \bis x_2$, 
\begin{description}
\item[(atom)] 
$\at_{\mathfrak{M}_1}^{\varrho}(x_1) =\at_{\mathfrak{M}_2}^{\varrho}(x_2)$\textup{;}

\item[(move)] if $x_1R_1y_1$, then there is $y_2$ such that $x_2R_2y_2$ and $y_1 \bis y_2$; and, conversely, if $x_2R_2 y_2$, then there is $y_1$ with $x_1R_1 y_1$ and $y_1 \bis y_2$.
\end{description}
If there is such $\bis$ with $z_1\bis z_2$, we write $\mathfrak{M}_1,z_1 \simr \mathfrak{M}_2,z_2$. 
The following characterisation of bisimulations between descriptive models in terms of types is well-known; see~\cite{goranko20075} and references therein: 

%
%

\begin{lemma}\label{bisim-lemma}
%
%
For any $\varrho \subseteq \sigma$, descriptive $\sigma$-models $\mathfrak{M}_i$, $i = 1,2$, and $x_i \in W_i$,
\[
 t_{\mathfrak{M}_1}^{\varrho}(x_1)=t_{\mathfrak{M}_2}^{\varrho}(x_2)\quad \text{ iff } \quad \mathfrak{M}_1,x_1  \simr \mathfrak{M}_2,x_2.
\]
The implication $(\Leftarrow)$ holds for arbitrary \textup{(}not necessarily descriptive\textup{)} models.
\end{lemma}

Variations of the next criterion of interpolant (non-)existence are implicit in various (dis-)proofs of the CIP in modal logics~\cite{DBLP:conf/amast/Marx98,goranko20075}: 

\begin{lemma}\label{criterion}
Let $\sigma = \sig(\varphi) \cup \sig(\psi)$ and $L \in \{\wKF, \DL\}$. Then $\varphi \to \psi$ has no interpolant in $L$ iff there are descriptive $\sigma$-models $\mathfrak M_\varphi$ and $\mathfrak M_\psi$ based on frames for $L$ with roots $\rr_\varphi$ and $\rr_\psi$, respectively, such that
\begin{description}
\item[$(a)$] $\mathfrak{M}_{\varphi},\rr_{\varphi} \models \varphi$ and $\mathfrak{M}_{\psi},\rr_{\psi} \models \neg\psi$\textup{;}


\item[$(b)$] $\mathfrak{M}_{\varphi},\rr_{\varphi} \simr \mathfrak{M}_{\psi},\rr_{\psi}$, where $\varrho = \sig(\varphi) \cap \sig(\psi)$.
\end{description}
As both $\wKF$ and $\DL$ are canonical, the requirement that models $\mathfrak M_\varphi$ and $\mathfrak M_\psi$ be descriptive can be omitted.
\end{lemma}
 
We first apply this criterion to decide the IEP for the difference logic $\DL$.
A key observation is that, from any two 
$\sigma$-models $\mathfrak M_\varphi$ and $\mathfrak M_\psi$ witnessing the nonexistence of an interpolant for a given $\varphi \to \psi$ in $\DL$ in the sense of Lemma~\ref{criterion}, we can extract sub-models of polynomial size in $|\varphi|$ and $|\psi|$ that also satisfy the above criterion. We call this phenomenon the \emph{polysize bisimilar model property} of $\DL$, which clearly implies that the IEP for $\DL$ is decidable in \textsc{coNP}. Indeed, to check that $\varphi \to \psi$ has no interpolant in $\DL$, we can guess polynomial-size $\mathfrak M_\varphi$ and $\mathfrak M_\psi$ together with a relation $\bis$ between them and then check whether they satisfy the criterion of Lemma~\ref{criterion}.

We remind the reader that rooted frames for $\DL$ (and so the frames $\mathfrak M_\varphi$ and $\mathfrak M_\psi$ are based on) are clusters, containing possibly both reflexive and irreflexive points. 
To show the polysize bisimilar model property of $\DL$, we proceed in two steps. First, for every $\alpha \in \sub(\varphi)$ ($\alpha \in \sub(\psi)$) satisfiable in $\mathfrak M_\varphi$ (respectively, $\mathfrak M_\psi$), we pick two points $x_\alpha,x_\alpha'$ satisfying $\alpha$ in $\mathfrak M_\varphi$ (in $\mathfrak M_\psi$) if they exist, otherwise a single such point $x_\alpha$. 
Denote the set of the points selected this way by $M_\varphi$ ($M_\psi$), assuming that $\rr_\varphi \in M_\varphi$ and $\rr_\psi \in M_\psi$. Let
%
\[
T = \{ t^\varrho_{\mathfrak M_\varphi}(x) \mid x \in M_\varphi \} \cup \{ t^\varrho_{\mathfrak M_\psi}(x) \mid x \in M_\psi \}.
\]
As $\mathfrak{M}_{\varphi},\rr_{\varphi} \simr \mathfrak{M}_{\psi},\rr_{\psi}$,
every $\varrho$-type $\type \in T$ is satisfied in both $\mathfrak M_\varphi$ and $\mathfrak M_\psi$. Now, for each $\type \in T$, 
we pick two distinct points satisfying $\type$ in $\mathfrak M_\varphi$, if they exist, and otherwise a single such point, and add them to $M_\varphi$ if they were not already there. We do the same for $\mathfrak M_\psi$ and $M_\psi$. Let $\mathfrak M^\dag_\varphi$ and $\mathfrak M^\dag_\psi$ be the restrictions of $\mathfrak M_\varphi$ and $\mathfrak M_\psi$ to the resulting $M_\varphi$ and $M_\psi$, and let 
%
\[
\bis^\dag=\{(x,x')\in M_\varphi \times M_\psi\mid t^\varrho_{\mathfrak M_\varphi}(x) =t^\varrho_{\mathfrak M_\psi}(x') \}.
\]

\begin{lemma}\label{l:dl}\
$(a)$ $\mathfrak{M}^\dag_{\varphi},\rr_{\varphi} \models \varphi$, $\mathfrak{M}^\dag_{\psi},\rr_{\psi} \models \neg\psi$,  and 
$(b)$ 
$\bis^\dag$ is a $\varrho$-bisimulation between $\mathfrak M^\dag_\varphi$ and $\mathfrak M^\dag_\psi$ with $\rr_{\varphi}\bis^\dag \rr_\psi$. 
\end{lemma}
%
\begin{proof}
Suppose that $\mathfrak M_\varphi$ and $\mathfrak M_\psi$ are based on the respective clusters 
$(W_\varphi,R_\varphi)$ and $(W_\psi,R_\psi)$. 

$(a)$ $\mathfrak{M}^\dag_{\varphi},\rr_{\varphi} \models \varphi$ follows from the fact that, for any $\chi\in\sub(\varphi)$ and $x\in M_\varphi$,
$\mathfrak M_\varphi,x\models\chi$ iff $\mathfrak M_\varphi^\dag,x\models\chi$, which can be established by a straightforward induction on the construction of $\varphi$.
We only show $(\Rightarrow)$ for $\chi=\Diamond\alpha$. 
Suppose that $\Diamond\alpha\in\sub(\varphi)$, $x\in M_\varphi$, and $\mathfrak M_\varphi,x\models\Diamond\alpha$. If there exist two points in $\mathfrak M_\varphi$ satisfying $\alpha$, then either $x_\alpha$ or $x_\alpha'$ is distinct from $x$, and so 
$\mathfrak M^\dag_\varphi,x\models\Diamond\alpha$. Otherwise, $xR_\varphi x_\alpha$ must hold, and so we also have $\mathfrak M^\dag_\varphi,x\models\Diamond\alpha$. 
(A similar argument shows that $\mathfrak{M}^\dag_{\psi},\rr_{\psi} \models \neg\psi$.)

$(b)$
Suppose $x,y\in M_\varphi$, $x'\in M_\psi$, $x\bis^\dag x'$, and $xR_\varphi y$.
Then $t^\varrho_{\mathfrak M_\varphi}(x)=t^\varrho_{\mathfrak M_\psi}(x')$ and $t^\varrho_{\mathfrak M_\varphi}(y)\in T$.
There are two cases: 
$(i)$ If $t^\varrho_{\mathfrak M_\varphi}(x)\ne t^\varrho_{\mathfrak M_\varphi}(y)$, then there is $y'\in M_\psi$ such that 
$t^\varrho_{\mathfrak M_\varphi}(y)=t^\varrho_{\mathfrak M_\varphi}(y')$ and $y'\ne x'$.
$(ii)$ If $t^\varrho_{\mathfrak M_\varphi}(x)=t^\varrho_{\mathfrak M_\varphi}(y)=\type$, then $\Diamond\type\subseteq\type$,
and so either $(ii.1)$ there are two distinct points in $M_\psi$ satisfying $\type$ in $\mathfrak M_\psi$ or $(ii.2)$ there is a single reflexive point in $M_\psi$ satisfying $\type$ in $\mathfrak M_\psi$.
In case $(ii.1)$ one of these two points must be different from $x'$, and in case $(ii.2)$ this single reflexive point must be $x'$. Thus, in 
all cases, 
we have a point $y'\in M_\psi$ with $x'R_\psi y'$ and 
$t^\varrho_{\mathfrak M_\psi}(y)=t^\varrho_{\mathfrak M_\psi}(y')$, 
and so $y\bis^\dag y'$. 
\end{proof}

By the construction, $|M_\varphi|$ and $|M_\psi|$ are polynomial in $|\varphi|$ and $|\psi|$. Thus, we obtain:

\begin{theorem}
$(a)$ $\DL$ enjoys the polysize bisimilar model property.

$(b)$ The interpolant existence property for $\DL$ is \emph{\textsc{coNP}}-complete. 
\end{theorem}


\section{Deciding Interpolant Existence for $\wKF$}\label{sec:deciding}

In this section we 
show that $\wKF$ has the 3-exponential-size bisimilar model property, which means that the IEP is decidable in \textsc{coN3ExpTime}. 

Given formulas $\varphi$ and $\psi$, let $\sub(\varphi,\psi) = \sub(\varphi) \cup \sub(\psi)$, $\sigma = \sig(\varphi) \cup \sig(\psi)$ and $\varrho = \sig(\varphi) \cap \sig(\psi)$. If $\varphi \to \psi$ does not have an interpolant in $\wKF$, then
Lemma~\ref{criterion} provides two pointed descriptive $\sigma$-models $\mathfrak M_\varphi,\rr_\varphi$ and $\mathfrak M_\psi,\rr_\psi$  based on weakly transitive frames.
To simplify notation, we will operate with a single descriptive $\sigma$-model $\mathfrak M$ based on a weakly transitive frame $\mathfrak{F} = (W,R)$---the disjoint union of $\mathfrak M_\varphi$ and $\mathfrak M_\psi$---containing two points $\rr_\varphi$, $\rr_\psi$ such that $\mathfrak{M},\rr_{\varphi} \models \varphi$, $\mathfrak{M},\rr_{\psi} \models \neg\psi$ and $\mathfrak M,\rr_\varphi \simr \mathfrak M,\rr_\psi$. 
Our aim is to convert $\mathfrak M$ into a model $\mathfrak M^\dag$ based on a weakly transitive frame $\mathfrak F^\dag = (W^\dag,R^\dag)$ that still witnesses the lack of an interpolant for $\varphi \to \psi$ in the above sense,
and has $W^\dag$ of triple-exponential size in $|\sub(\varphi,\psi)|$.

Given a point $x$ in $\mathfrak M$, we define the $\varphi,\psi$-\emph{type} $\ftype^{\varphi,\psi}_{\mathfrak{M}}(x) = t^\sigma_{\mathfrak M}(x) \cap \sub(\varphi, \psi)$.
For a set $X$ of points in $\mathfrak M$, we let $\ftype_{\mathfrak{M}}^{\varphi,\psi}(X)=\bigl\{\ftype_{\mathfrak{M}}^{\varphi,\psi}(x)\mid x\in X\bigr\}$. 
Our construction 
is an elaborate $\sub(\varphi,\psi)$-filtration.
(For the well-known \emph{filtration} techniques in modal logic see, e.g., \cite{DBLP:books/daglib/0030819}.) 
As usual, it keeps track of the $\varphi,\psi$-types of points in $\mathfrak M$
to satisfy condition $(a)$ of Lemma~\ref{criterion}. In addition, some connections of these $\varphi,\psi$-types with the
$\varrho$-types of points in $\mathfrak M$ are also noted in order to satisfy condition $(b)$ of Lemma~\ref{criterion}.
When applied to satisfiability checking, our construction reduces to the filtration of \cite{kudinov2024filtrations}; see Remark~\ref{r:filtr}.
%

To begin with, we
define an equivalence relation $\equC$ between clusters in $\mathfrak{M}$ by taking $C \equC C'$ 
iff there exists a sequence $C = C_{0},\dots, C_{n} = C'$ of clusters in $\mathfrak{M}$ such that, for each $i < n$ there are $x_{i}\in C_{i}$ and $y_{i+1}\in C_{i+1}$ with $t_{\mathfrak M}^\varrho(x_i) = t_{\mathfrak M}^\varrho(y_{i+1})$.
Let $\cC{C} = \{ C' \mid C'\equC C\}$ and $\TC{C} = \bigcup_{C' \in\cC{C}} t^\varrho_{\mathfrak M}(C')$. 
It follows from the definition that $\cC{C} = \{C' \mid t_{\mathfrak{M}}^{\varrho}(C') \cap \TC{C} \ne \emptyset \}$.
By 
Lemma~\ref{l:fundament}~$(d)$, 
for every $C'\in \cC{C}$,  there is a $\varrho$-maximal $D\in\cC{C}$ with $C'R^{r}D$ and 
$t^\varrho_{\mathfrak M}(C')\subseteq t^\varrho_{\mathfrak M}(D)$.

\begin{lemma}\label{rho-clusters}
If $D \in \cC{C}$ is a $\varrho$-maximal cluster, then 
$\TC{C} = t_{\mathfrak{M}}^{\varrho}(D)$, and so $|\TC{C}| \le 2^{|\varrho|} $.
%
%
\end{lemma}
\begin{proof}
Let $D \in \cC{C}$ be a $\varrho$-maximal cluster. Clearly, it is enough to prove that $\TC{C} \subseteq t_{\mathfrak{M}}^{\varrho}(D)$.
Let $x \in C'$, for some $C'\in\cC{C}$. 
We need to show that there exists $y \in D$ with  $t_{\mathfrak{M}}^{\varrho}(y) = t_{\mathfrak{M}}^{\varrho}(x)$. 
As $D\equC C'$, by the definition of $\equC$, we have $z \in D$ 
such that $\Diamond^n t_{\mathfrak{M}}^{\varrho}(z)\subseteq t_{\mathfrak{M}}^{\varrho}(x)$ and
$\Diamond^n t_{\mathfrak{M}}^{\varrho}(x)\subseteq t_{\mathfrak{M}}^{\varrho}(z)$,
for some $n \ge 0$. If $n = 0$, then we take $y = z$. If $n>0$, then by 
repeated applications of Lemma~\ref{l:fundament}~$(c)$, we obtain 
$u,v \in W$ with $z R^n u R^n v$, 
$t_{\mathfrak{M}}^{\varrho}(u) = t_{\mathfrak{M}}^{\varrho}(x)$ and $t_{\mathfrak{M}}^{\varrho}(v) = t_{\mathfrak{M}}^{\varrho}(z)$. 
By weak transitivity, $z = v$ or $z R v$. In either case, by the $\varrho$-maximality of $D$, we must have $v \in D$, and so $u \in D$, yielding $y = u$.
%
\end{proof}


%

For every equivalence class $\cC{C}$ in $\mathfrak M$, we let $\ATC{C}=\{\type\cap\varrho\mid\type\in\TC{C}\}$.
Given a cluster $C$, we define the \emph{cluster-type} of $C$ in $\mathfrak{M}$ as the function
$\Type_C\colon\ATC{C}\to 2^{\ftype_{\mathfrak{M}}^{\varphi,\psi}(C)}$ where, for any $\atype\in\ATC{C}$, 
%
\[
\Type_C(\atype)=\{ \ftype_{\mathfrak{M}}^{\varphi,\psi}(x) \mid x\in C, \ \ftype_{\mathfrak{M}}^{\varphi,\psi}(x)\cap\varrho=\at^{\varrho}_{\mathfrak{M}}(x)=\atype \}.
\]
Observe that 
$\ATC{C'}=\ATC{C}$ for every $C'\in\cC{C}$, and so $\Type_C$ and $\Type_{C'}$ have the same domain $\dom\,\Type_C=\dom\,\Type_{C'}=\ATC{C}$.
As $\bigcup_{\atype\in\dom\,\Type_C}\Type_C(a)=\ftype^{\varphi,\psi}_{\mathfrak M}(C)$,
$\Type_C$ keeps a record of both $\varphi,\psi$-types and atomic $\varrho$-types of points in $C$ in the context of the whole equivalence class $\cC{C}$. 
Also, $\Type_C(\atype)$ might be empty for some $\atype\in\dom\,\Type_C$, but if $C$ is $\varrho$-maximal then 
$\Type_C(\atype)\ne\emptyset$ for all $\atype\in\dom\,\Type_C$.
Note that the number of pairwise distinct cluster-types in $\mathfrak M$ does not exceed 
$2^{|\varrho|}\cdot \bigl(2^{2^{|\sub(\varphi,\psi)|}}\bigr)^{2^{|\varrho|}}=2^{|\varrho|+2^{|\sub(\varphi,\psi)|+|\varrho|}}$.

By the \emph{mosaic of $\cC{C}$ in} $\mathfrak M$ we mean the set 
$\MM{C} = \bigl\{ \Type_{C'} \mid C' \in \cC{C}\bigr\}$ of cluster-types of the same 
domain (also called as the \emph{domain of} $\MM{C}$ and denoted by $\dom\,\MM{C}$).
$M$ is a \emph{mosaic in} $\mathfrak M$ if $M=\MM{C}$ for some $C$.
Clearly,  the number of pairwise distinct mosaics in $\mathfrak M$ is $\mathcal{O}\bigl(2^{2^{2^{|\sub(\varphi,\psi)|}}}\bigr)$.

We are now in a position to define the model $\mathfrak M^\dag = (\mathfrak F^\dag, \mathfrak v^\dag)$ and its underlying frame $\mathfrak F^\dag = (W^\dag,R^\dag)$. Suppose $x \in W$.
Then we set
\[
\p(x) = \bigl(\ftype_{\mathfrak{M}}^{\varphi,\psi}(x) ,\at^\varrho_{\mathfrak M}(x),\Type_{C(x)},\MM{C(x)}\bigr) \quad \text{and} \quad W^\dag = \{ \p(x) \mid x \in W\}.
\]
Observe that if $\p\in W^\dag$ and $\p = (\ftype,\atype,\Type, M)$, then
$\atype=\ftype\cap\varrho$, $\Type\in M$ and $\ftype\in\Type(\atype)$ always hold.
Moreover,
\begin{equation}\label{abstract}
\mbox{$(\ftype,\atype,\Type, M)\in W^\dag$, for all mosaics $M$, $\Type\in M$, $\atype\in\dom\, M$, and $\ftype\in\Type(\atype)$.}
\end{equation}
We call $M$ and $\atype$ the \emph{mosaic} and the $\varrho$-\emph{index} of $\p$, respectively.
 Later on, we shall see that $\p$ and $\p'$ are $\varrho$-bisimilar if they share the same mosaic and $\varrho$-index. 
%
%
%
By the above calculations,  $|W^\dag| =\mathcal{O}\bigl(2^{2^{2^{|\sub(\varphi,\psi)|}}}\bigr)$.
The valuation $\mathfrak v^\dag$ on $W^\dag$ is inherited from $\mathfrak v$ in $\mathfrak M$: for every $p \in \sigma$,
\[
\mathfrak v^\dag(p) = \{ \p(x) \mid x \in \mathfrak v(p) \} = \{ \p(x) \mid p \in t_{\mathfrak{M}}^{\varphi,\psi}(x), \ x \in W \}.
\]
%
%
To define the accessibility relation $R^\dag$ in $\mathfrak F^\dag$, we require some new notions. 
Let $\ftype$ and $\ftype'$ be $\varphi,\psi$-types, 
$\Type_C\colon\ATC{C}\to 2^{\ftype_{\mathfrak{M}}^{\varphi,\psi}(C)}$ and 
$\Type_{C'}\colon\ATC{C'}\to 2^{\ftype_{\mathfrak{M}}^{\varphi,\psi}(C')}$
cluster-types, $M$ and $M'$ mosaics. Define a relation $\twoheadrightarrow$ between such pairs by taking:
\begin{itemize}
\item[--] $\ftype \twoheadrightarrow  \ftype'$ iff, for every $\Diamond \psi \in \sub(\varphi, \psi)$, whenever $\chi$ or $\Diamond\chi$ is in $\ftype'$, then $\Diamond\chi\in \ftype$;

\item[--] $\Type_C \twoheadrightarrow \Type_{C'}$ iff $\ftype\twoheadrightarrow\ftype'$ for all $\ftype\in \ftype^{\varphi,\psi}_{\mathfrak M}(C)$ and
$\ftype'\in \ftype^{\varphi,\psi}_{\mathfrak M}(C')$;
%

\item[--] $M \twoheadrightarrow M'$ iff for every $\Type\in M$ there is $\Type'\in M'$ with $\Type \twoheadrightarrow \Type'$.
\end{itemize}
It is readily checked that the defined relation $\twoheadrightarrow$ has the following properties:
\begin{align}
\label{trans}
& \mbox{$\twoheadrightarrow$ is transitive in all three settings;}\\
\label{rto}
& \mbox{if $xRy$, then $\ftype_{\mathfrak{M}}^{\varphi,\psi}(x) \twoheadrightarrow \ftype_{\mathfrak{M}}^{\varphi,\psi}(y)$;}\\
\label{crc}
& \mbox{if $CRC'$ and $C\ne C'$, then $\Type_{C} \twoheadrightarrow \Type_{C'}$.}
\end{align}
Note that $CR C$ does not necessarily imply $\Type_{C} \twoheadrightarrow \Type_{C}$; see Example~\ref{ex:basic3}.
By \eqref{crc}, Lemmas \ref{l:fundament}~$(d)$ and \ref{rho-clusters}, we also have that
\begin{equation}\label{tomax}
\mbox{for all $C$ there is a $\varrho$-maximal $D\in\cC{C}$ such that $C=D$ or $\Type_C\twoheadrightarrow \Type_{D}$.}
\end{equation}

Next, for any mosaic $M$, we define a subset $\iset{M}$ of $\dom\, M$ by taking
\begin{align*}
\iset{M}=\{\atype\in\dom\, M\mid \Diamond\type\not\subseteq\type,\  \mbox{for all clusters $C$} & \mbox{ with $M=\MM{C}$ and}\\
& \mbox{ all $\type\in\TC{C}$ with $\type\cap\varrho=\atype$} \}.
\end{align*}
As $\atype\in\iset{M}$ implies that for every cluster $C$ with $M=\MM{C}$ there is at most one $x\in C$ such that $\ftype^{\varphi,\psi}_{\mathfrak M}(x)\in\Type_C(\atype)$ and such an $x$ is irreflexive, it follows that
\begin{equation}\label{singleton}
\mbox{if $\Type_C\in M$, then $|\Type_C(\atype)|\leq 1$ for every $\atype\in\iset{M}$.}
\end{equation}
We also claim that
\begin{equation}\label{irrall}
\mbox{if $C$ is not $\varrho$-maximal and $\Type_C\in M$, then $\Type_C(\atype)=\emptyset$ for every $\atype\in\iset{M}$.}
\end{equation}
Indeed, by Lemmas~\ref{l:fundament}~$(d)$ and \ref{rho-clusters}, there is a $\varrho$-maximal $D\in [C(x)]$ with $C\ne D$ and $CRD$.
Let $\atype\in\dom\, \MM{C}$ be such that $\Type_C(\atype)\ne\emptyset$. There there is $x\in C$ with $t^\varrho_{\mathfrak M}(x)\cap\varrho=\atype$. By Lemma~\ref{rho-clusters}, there is $y\in D$ with $t^\varrho_{\mathfrak M}(y)=t^\varrho_{\mathfrak M}(x)$. Thus,
$\Diamond t^\varrho_{\mathfrak M}(x)\subseteq t^\varrho_{\mathfrak M}(x)$, and so $\atype\notin\iset{\MM{C}}$.


%
%
%

%

%

Now, to define $R^\dag$ on $W^\dag$, suppose $\p,\p' \in W^\dag$, where 
\[
\p = (\ftype,\atype,\Type, M), \qquad \p' = (\ftype',\atype',\Type', M').
\]
The definition of $R^\dag$ for $\p$, $\p'$ depends on whether $M=M'$ and $M \twoheadrightarrow M$:
%
%
%
\begin{description}
\item[\rm\textit{Case} $M \ne M'$:] then $\p R^\dag \p'$ iff $M \twoheadrightarrow M'$ and $\Type \twoheadrightarrow \Type'$.

\item[\rm\textit{Case} $M = M'$, $M \twoheadrightarrow M$:] then $\p R^\dag \p'$ iff 
\begin{itemize}
\item[--] either $\p \ne \p'$ and ($\Type = \Type'$ or $\Type \twoheadrightarrow \Type'$), 

\item[--] or $\p = \p'$ and $\ftype \twoheadrightarrow \ftype$.
\end{itemize}

\item[\rm\textit{Case} $M = M'$, $M \not\twoheadrightarrow M$:] 
%
%
then $\p R^\dag \p'$ iff 
\begin{itemize}
\item[--] either $\p \ne \p'$ and $\bigl(\Type = \Type'$ or $(\Type \twoheadrightarrow \Type'$ and $\Type(\btype)=\emptyset$ for all $\btype\in\iset{M})\bigr)$,
	
\item[--] or $\p = \p'$, $\ftype \twoheadrightarrow \ftype$, and $\atype\not\in\iset{M}$.
\end{itemize}
\end{description}		
%
%

\begin{example}\label{ex:basic3}
Consider $\varphi$, $\psi$, $\mathfrak M_\varphi$ and $\mathfrak M_\psi$ with $\sigma = \{p,q\}$ and $\varrho = \emptyset$ from Example~\ref{ex:basic} $(i)$. 
Then $\sub(\varphi,\psi)$ consists of the formulas $p$, $q$, $\Diamond p$, $\Diamond\neg q$, $\Diamond\Diamond p$, $\Diamond\Diamond\neg q$, $\varphi$, $\psi$ and their negations.
Let $\mathfrak M$ be the disjoint union of $\mathfrak M_\varphi$ and $\mathfrak M_\psi$. Then $C(\rr_\varphi) = \{\rr_\varphi,x_1\} = C(x_1)$, $C(\rr_\psi) = \{\rr_\psi\}$, $C(x_2) = \{x_2\}$, with all of these clusters being $\equC$-equivalent, 
$\TC{C(\rr_\varphi)} = \{\type\}$ for  $\type = t^\varrho_{\mathfrak M}(\rr_\varphi) = \{\Diamond^n \top \mid n < \omega\}$,
and $\dom\,\Type_{C(\rr_\varphi)}=\dom\,\Type_{C(\rr_\psi)}=\dom\,\Type_{C(x_2)}=\ATC{C(\rr_\varphi)} = \{\emptyset\}$. This gives 
$\Type_{C(\rr_\varphi)}(\emptyset)= \{\ftype^{\varphi,\psi}_{\mathfrak M}(\rr_\varphi), \ftype^{\varphi,\psi}_{\mathfrak M}(x_1)\}$, 
where
%
\begin{align*}
& \ftype^{\varphi,\psi}_{\mathfrak M}(\rr_\varphi) = \{p, \neg q, \neg\Diamond p, \Diamond \neg q, \Diamond\Diamond p, \Diamond\Diamond\neg q, \varphi, \psi \}, \\ 
& \ftype^{\varphi,\psi}_{\mathfrak M}(x_1) = \{\neg p, \neg q, \Diamond p, \Diamond\neg q, \Diamond\Diamond p, \Diamond\Diamond \neg q, \neg \varphi, \psi \}.
\end{align*}
On the other hand, $\Type_{C(\rr_\psi)}(\emptyset) = \{\ftype^{\varphi,\psi}_{\mathfrak M}(\rr_\psi)\}$ and $\Type_{C(x_2)}(\emptyset) = \{\ftype^{\varphi,\psi}_{\mathfrak M}(x_2)\}$, where
\begin{align*}
& \ftype^{\varphi,\psi}_{\mathfrak M}(\rr_\psi) = \{\neg p, \neg q, \neg\Diamond p, \neg\Diamond\neg q, \neg\Diamond \Diamond p, 
 \neg\Diamond \Diamond\neg q, \neg \varphi, \neg \psi \}, \\ 
& \ftype^{\varphi,\psi}_{\mathfrak M}(x_2) = \{\neg p, q, \neg\Diamond p, \neg\Diamond\neg q, \neg\Diamond \Diamond p, 
 \neg\Diamond \Diamond\neg q, \neg \varphi, \psi \}.
\end{align*}
%
Then $\p(\rr_\varphi)$, $\p(x_1)$, $\p(\rr_\psi)$, and $\p(x_2)$ are all different, but they share the same $\varrho$-index $\emptyset$ and the same mosaic
$\MM{C(\rr_\varphi)} = \{\Type_{C(\rr_\varphi)}, \ \Type_{C(\rr_\psi)}, \Type_{C(x_2)}\}$
with $\dom\,\MM{C(\rr_\varphi)} = \{\emptyset\}$. As $\Diamond\type\subseteq\type$, we have $\iset{\MM{C(\rr_\varphi)}} = \emptyset$. 

We have $C(\rr_\varphi) R C(\rr_\varphi)$ but $\Type_{C(\rr_\varphi)} \not\twoheadrightarrow \Type_{C(\rr_\varphi)}$ as 
$\ftype^{\varphi,\psi}_{\mathfrak M}(\rr_\varphi) \not \twoheadrightarrow \ftype^{\varphi,\psi}_{\mathfrak M}(\rr_\varphi)$
because $p \in \ftype^{\varphi,\psi}_{\mathfrak M}(\rr_\varphi)$ but $\Diamond p \notin \ftype^{\varphi,\psi}_{\mathfrak M}(\rr_\varphi)$.
So $\p(\rr_\varphi) R^\dag \p(\rr_\varphi)$ does not hold. 
In fact, it is not hard to check that $\mathfrak M^\dag$ is isomorphic to $\mathfrak M$. For example,
we do not have $\p(\rr_\psi) R^\dag \p(\rr_\psi)$ as $\ftype^{\varphi,\psi}_{\mathfrak M}(\rr_\psi) \not \twoheadrightarrow \ftype^{\varphi,\psi}_{\mathfrak M}(\rr_\psi)$ because $\neg q \in \ftype^{\varphi,\psi}_{\mathfrak M}(\rr_\psi)$ but $\Diamond \neg q\notin \ftype^{\varphi,\psi}_{\mathfrak M}(\rr_\psi)$. But we
obtain $\p(\rr_\varphi) R^\dag \p(x_1)$ because they share the same cluster-type $\Type_{C(\rr_\varphi)}$, and $\p(x_1) R^\dag \p(x_1)$ because $\ftype^{\varphi,\psi}_{\mathfrak M}(x_1) \twoheadrightarrow \ftype^{\varphi,\psi}_{\mathfrak M}(x_1)$.
\hfill$\dashv$
\end{example}

\begin{example}\label{ex:irrset}
We now illustrate the role of the set $\iset{M}$ in the definition of $R^\dag$. Consider the model $\mathfrak M$ below and $\varphi = p \land \neg \Diamond q$ and $\psi = \neg p \land \Diamond r$ with $\varrho = \{p\}$ and $\sigma = \{p, q, r\}$. (The form of $\varphi$, $\psi$ is not important here, but $\sub(\varphi,\psi)$ is.)
%
%
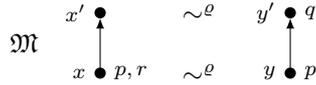
\begin{center}
\begin{tikzpicture}[>=latex,line width=0.4pt,xscale = 1,yscale = 1]
\node at (-1,0) {$\mathfrak M$};
\node[point,fill=black,scale = 0.6,label=left:{\footnotesize $x'$}] (f1) at (0,.4) {};
\node[point,fill=black,scale = 0.6,label=left:{\footnotesize $x$},label=right:{\footnotesize $p,r$}] (f0) at (0,-.4) {};
\draw[->] (f0) to (f1);
\node[point,fill=black,scale = 0.6,label=left:{\footnotesize $y'$},label=right:{\footnotesize $q$}] (p1) at (2.5,.4) {};
\node[point,fill=black,scale = 0.6,label=left:{\footnotesize $y$},label=right:{\footnotesize $p$}] (p0) at (2.5,-.4) {};
\draw[->] (p0) to (p1);
\node at (1.3,-.4) {$\simr$};
\node at (1.3,.4) {$\simr$};
\end{tikzpicture}
\end{center}
Then $t^\varrho_{\mathfrak M}(x) = t^\varrho_{\mathfrak M}(y) = \type$ with $\Diamond \type \not \subseteq \type$. Let $\ftype^{\varphi,\psi}_{\mathfrak M}(x) = \ftype_x$ and $\ftype^{\varphi,\psi}_{\mathfrak M}(y) = \ftype_y$. Then $\ftype_x \not \twoheadrightarrow \ftype_x$ because 
$r \in \ftype^{\varphi,\psi}_{\mathfrak M}(x)$ and $\Diamond r \notin \ftype^{\varphi,\psi}_{\mathfrak M}(x)$,
while $\ftype_y \twoheadrightarrow \ftype_y$ because $\Diamond p \notin \sub(\varphi,\psi)$; we also have $\ftype_x \not \twoheadrightarrow \ftype_y$ and $\ftype_y \not \twoheadrightarrow \ftype_x$. Now, consider 
\[
\p_x = (\ftype_x,\{p\},\Type_{C(x)}, M), \quad 
\p_y = (\ftype_y,\{p\},\Type_{C(y)}, M),
\]
where $\TC{C(x)} = \{t\}$, $\cC{C(x)} = \{C(x), C(y)\}$, and $M = \{\Type_{C(x)}, \Type_{C(y)}\}$ with $\dom\,M=\ATC{C(x)}=\bigl\{\{p\}\bigr\}$,
$\Type_{C(x)}(\{p\}) = \{\ftype_x\}$, and $\Type_{C(y)}(\{p\}) = \{\ftype_y\}$.
Then $M \not \twoheadrightarrow M$ and $\iset{M}=\bigl\{\{p\}\bigr\}$.
By the last item in the definition of $R^\dag$, neither $\p_x R^\dag \p_x$ nor $\p_y R^\dag \p_y$ holds because $\{p\}\in\iset{M}$.
However, without the condition $\atype\notin\iset{M}$ in the definition,
we would have $\p_y R^\dag \p_y$ but still not $\p_x R^\dag \p_x$, which would destroy the $\varrho$-bisimilarity of $\p_x$ and $\p_y$.
\hfill$\dashv$
\end{example}

%
\begin{lemma}\label{l:Rwtrans}
The relation $R^\dag$ on $W^\dag$ is weakly transitive. 
\end{lemma}

\begin{proof}
Suppose $\p R^\dag \p' R^\dag \p''$ and $\p \ne \p''$, where 
%
%
\[
\p=(\ftype,\atype,\Type, M),\quad \p'=(\ftype',\atype',\Type', M'),\quad \p''=(\ftype'',\atype'',\Type'', M'').
\]
We need to show,
by a straightforward checking of all the cases in the definition of $R^\dag$, that $\p R^\dag \p''$. This  is trivial if $\p = \p'$ or $\p' = \p''$, so we assume that $\p \ne \p'$ and $\p' \ne \p''$.

Suppose first that $M \ne M''$. We need to show that $M \twoheadrightarrow M''$ and $\Type \twoheadrightarrow \Type''$.
\begin{description} 
\item[\rm\emph{Case} $M \ne M'$, $M' \ne M''$:] By the definition of $R^\dag$, we have $M \twoheadrightarrow M' \twoheadrightarrow M''$ and $\Type \twoheadrightarrow \Type' \twoheadrightarrow \Type''$, which gives $M \twoheadrightarrow M''$ and $\Type \twoheadrightarrow \Type''$ by \eqref{trans}.

\item[\rm\emph{Case} $M \ne M'$, $M'= M''$:] Then $M \twoheadrightarrow M'$ and $\Type \twoheadrightarrow \Type'$ because $\p R^\dag \p'$, and so $M \twoheadrightarrow M''$. Then no matter whether $M' \twoheadrightarrow M''$ or $M' \not\twoheadrightarrow M''$,
we have either $\Type' = \Type''$ or  $\Type' \twoheadrightarrow \Type''$ as $\p' R^\dag \p''$ and $\p'\ne\p''$.
Thus, $\Type \twoheadrightarrow \Type''$ by \eqref{trans}.

%

\item[\rm\emph{Case} $M = M'$, $M' \ne M''$:] Then $M' \twoheadrightarrow M''$ and $\Type' \twoheadrightarrow\Type''$ because $\p' R^\dag \p''$, and so $M \twoheadrightarrow M''$. 
Then no matter whether $M\twoheadrightarrow M'$ or $M\not\twoheadrightarrow M'$,
either $\Type= \Type'$ or  $\Type\twoheadrightarrow \Type'$, because $\p R^\dag \p'$ and $\p\ne\p'$.
Thus, $\Type \twoheadrightarrow \Type''$ by \eqref{trans}.
%
\end{description}

Now suppose $M = M''$. If $M \ne M'$, then $M' \ne M''$. By the definition of $R^\dag$, we have $M \twoheadrightarrow M' \twoheadrightarrow M$ and $\Type \twoheadrightarrow \Type' \twoheadrightarrow \Type''$, which gives $M \twoheadrightarrow M$ and $\Type \twoheadrightarrow \Type''$ by \eqref{trans}. As $\p\ne\p''$, these imply $\p R^\dag\p''$.

Finally, suppose $M = M' = M''$. Then two cases are possible.
\begin{description} 
\item[\rm\emph{Case} $M \twoheadrightarrow M$:] As $\p \ne \p''$, we need to show that $\Type = \Type''$ or $\Type \twoheadrightarrow \Type''$. As $\p R^\dag \p'$ and $\p \ne \p'$, we have $\Type = \Type'$ or $\Type \twoheadrightarrow \Type'$. Similarly, as $\p' R^\dag \p''$ and $\p' \ne \p''$, we have $\Type' = \Type''$ or $\Type' \twoheadrightarrow \Type''$, which yields the required.

\item[\rm\emph{Case} $M \not\twoheadrightarrow M$:] We need to show that $\Type = \Type''$ 
or $(\Type \twoheadrightarrow\Type''$ and $\Type(\btype)=\emptyset$ for all $\btype\in\iset{M})$.
As $\p R^\dag \p'$ and $\p\ne\p'$, we have $\Type = \Type'$ 
or $(\Type \twoheadrightarrow\Type'$ and $\Type(\btype)=\emptyset$ for all $\btype\in\iset{M})$.
Similarly, as $\p' R^\dag \p''$ and $\p'\ne\p''$, we have $\Type' = \Type''$ 
or $(\Type' \twoheadrightarrow\Type''$ and $\Type'(\btype)=\emptyset$ for all $\btype\in\iset{M})$.
By \eqref{trans}, it follows that $\Type = \Type''$ 
or $(\Type \twoheadrightarrow\Type''$ and $\Type(\btype)=\emptyset$ for all $\btype\in\iset{M})$. As $\p\ne\p''$, these imply $\p R^\dag\p''$.
\end{description}
This completes the proof of the lemma.
\end{proof}

The next lemma says that $R^\dag$ contains the smallest $\sub(\varphi,\psi)$-filtration:

\begin{lemma}\label{Rpreserved}
For all $x,y \in W$, if $xRy$, then $\p(x) R^\dag \p(y)$.
\end{lemma}
\begin{proof}
Suppose we have $xRy$.
To begin with, we claim that
\begin{equation}
\label{notequ}
 \mbox{if $\cC{C(x)}\ne\cC{C(y)}$ then $\MM{C(x)} \twoheadrightarrow \MM{C(y)}$.}
\end{equation}
Indeed, 
take any $\Type\in \MM{C(x)}$ and let $C\in \cC{C(x)}$ be such that $\Type_{C} =\Type$. By \eqref{tomax},
there is a $\varrho$-maximal $D\in \cC{C(x)}$ with $C=D$ or $\Type_C\twoheadrightarrow \Type_{D}$.
So, by Lemma~\ref{rho-clusters}, there is $x'\in D$ with $t^\varrho_{\mathfrak M}(x)=t^\varrho_{\mathfrak M}(x')$.
As $xRy$ and $\mathfrak M$ is descriptive, it follows from Lemma~\ref{bisim-lemma} that there is $y'$ with $x'R y'$ and $t^\varrho_{\mathfrak M}(y')=t^\varrho_{\mathfrak M}(y)$. Then $C(y')\in\cC{C(y)}\ne\cC{C(x)}=\cC{C(x')}=\cC{D}$. Therefore, 
$\Type_D\twoheadrightarrow\Type_{C(y')}$ by \eqref{crc}, and so $\Type_C\twoheadrightarrow\Type_{C(y')}$ by \eqref{trans}, as required.


Now we show that $\p(x) R^\dag \p(y)$ follows from \eqref{notequ} in all cases of the definition of $R^\dag$.
%
Assume first that $\MM{C(x)}\ne \MM{C(y)}$. Then $\cC{C(x)}\ne\cC{C(y)}$, and so $\MM{C(x)} \twoheadrightarrow \MM{C(y)}$ by \eqref{notequ},
and $\Type_{C(x)}\twoheadrightarrow \Type_{C(y)}$ by \eqref{crc}, as required.

	
Assume next that $\MM{C(x)}=\MM{C(y)}$ and $\MM{C(x)}\twoheadrightarrow \MM{C(y)}$. 
By \eqref{crc}, $\Type_{C(x)} \not\twoheadrightarrow \Type_{C(y)}$ implies $C(x)= C(y)$, and so $\Type_{C(x)}=\Type_{C(y)}$.
So we have $\p(x) R^\dag \p(y)$ when $\p(x)\ne\p(y)$.
 If $\p(x)=\p(y)$ then $\type_{\mathfrak{M}}^{\varphi,\psi}(x)\twoheadrightarrow \type_{\mathfrak{M}}^{\varphi,\psi}(y)$ follows by
 \eqref{rto}, and so we also have $\p(x) R^\dag \p(y)$.
%
%

Finally, assume that $\MM{C(x)}=\MM{C(y)}$, $\MM{C(x)}\not\twoheadrightarrow \MM{C(y)}$. 
Then \eqref{notequ} implies that $\cC{C(x)}=\cC{C(y)}$. There are two cases:
\begin{description}
\item[\rm\emph{Case} $C(x)$ is $\varrho$-maximal:] 
Then $C(x)=C(y)$, and so $\Type_{C(x)}=\Type_{C(y)}$. Thus, we have $\p(x)R^\dag\p(y)$ if $\p(x)\ne\p(y)$.
If $\p(x)=\p(y)$, then $\ftype_{\mathfrak{M}}^{\varphi,\psi}(x)\twoheadrightarrow  \ftype_{\mathfrak{M}}^{\varphi,\psi}(x)$ by 
\eqref{rto}. As $\at^\varrho_{\mathfrak M}(x)=\at^\varrho_{\mathfrak M}(y)$ and $C(x)=C(y)$, we have 
$t^\varrho_{\mathfrak M}(x)=t^\varrho_{\mathfrak M}(y)$, and so $\Diamond t^\varrho_{\mathfrak M}(x)\subseteq t^\varrho_{\mathfrak M}(x)$
by $xRy$. Therefore, $\at^\varrho_{\mathfrak M}(x)\notin\iset{\MM{C(x)}}$, and so $\p(x)R^\dag\p(y)$.

\item[\rm\emph{Case} $C(x)$ is not $\varrho$-maximal:] Then, by \eqref{irrall},
$\Type_{C(x)}(\atype)=\emptyset$ for every $\atype\in\iset{\MM{C(x)}}$. In particular,
as $\ftype_{\mathfrak{M}}^{\varphi,\psi}(x)\in\Type_{C(x)}(\at^\varrho_{\mathfrak M}(x))$, it follows that 
$\at^\varrho_{\mathfrak M}(x)\notin\iset{\MM{C(x)}}$.
If $\p(x)=\p(y)$ then we have $\ftype_{\mathfrak{M}}^{\varphi,\psi}(x)\twoheadrightarrow  \ftype_{\mathfrak{M}}^{\varphi,\psi}(x)$ by 
\eqref{rto}, and so $\p(x)R^\dag\p(y)$.
If $\p(x)\ne\p(y)$ and $C(x)=C(y)$, then $\Type_{C(x)}=\Type_{C(y)}$, and so $\p(x)R^\dag\p(y)$.
And if $\p(x)\ne\p(y)$ and $C(x)\ne C(y)$, then $\Type_{C(x)}\twoheadrightarrow\Type_{C(y)}$ by \eqref{crc}, and so  $\p(x)R^\dag\p(y)$ again.
\end{description}
%
%
%
This completes the proof of the lemma.
\end{proof}

The next lemma says that $R^\dag$ is contained in the largest $\sub(\varphi,\psi)$-filtration:

\begin{lemma}\label{l:maxfiltr}
If $\p(x)R^\dag\p(y)$, then $\chi\in\ftype^{\varphi,\psi}_{\mathfrak M}(y)$ implies 
$\Diamond\chi\in\ftype^{\varphi,\psi}_{\mathfrak M}(x)$, for every $\Diamond \chi \in \sub(\varphi, \psi)$.
\end{lemma}
\begin{proof}
As $\ftype_{\mathfrak{M}}^{\varphi,\psi}(x)\twoheadrightarrow \ftype_{\mathfrak{M}}^{\varphi,\psi}(y)$ implies that,
for every $\Diamond \chi \in \sub(\varphi, \psi)$, whenever $\chi\in\ftype^{\varphi,\psi}_{\mathfrak M}(y)$ then
$\Diamond\chi\in\ftype^{\varphi,\psi}_{\mathfrak M}(x)$, and $\Type_{C(x)}\twoheadrightarrow \Type_{C(y)}$ implies
$\ftype_{\mathfrak{M}}^{\varphi,\psi}(x)\twoheadrightarrow \ftype_{\mathfrak{M}}^{\varphi,\psi}(y)$, we only need to check those
cases where we have $\p(x)R^\dag\p(y)$ but neither $\ftype_{\mathfrak{M}}^{\varphi,\psi}(x)\twoheadrightarrow \ftype_{\mathfrak{M}}^{\varphi,\psi}(y)$
nor $\Type_{C(x)}\twoheadrightarrow \Type_{C(y)}$ holds. 

An inspection of the definition of $R^\dag$ shows that this can only happen when $\p(x)\ne\p(y)$, $\MM{C(x)}=\MM{C(y)}$ and 
$\Type_{C(x)}= \Type_{C(y)}$. In this case, we have
$\ftype_{\mathfrak{M}}^{\varphi,\psi}(x)\in\Type_{C(x)}\bigl(\at^\varrho_{\mathfrak M}(x)\bigr)=\Type_{C(y)}\bigl(\at^\varrho_{\mathfrak M}(x)\bigr)$, and so there is $y'\in C(y)$ with $\ftype_{\mathfrak{M}}^{\varphi,\psi}(y')=\ftype_{\mathfrak{M}}^{\varphi,\psi}(x)$.
Then $y'\ne y$, as otherwise we would have $\ftype_{\mathfrak{M}}^{\varphi,\psi}(x)=\ftype_{\mathfrak{M}}^{\varphi,\psi}(y)$,
and so $\at^\varrho_{\mathfrak M}(x)=\at^\varrho_{\mathfrak M}(y)$ as well, contradicting $\p(x)\ne\p(y)$.
Now it follows that, for every $\Diamond \chi \in \sub(\varphi, \psi)$, if $\chi\in\ftype^{\varphi,\psi}_{\mathfrak M}(y)$, then
$\Diamond\chi\in\ftype^{\varphi,\psi}_{\mathfrak M}(y')=\ftype^{\varphi,\psi}_{\mathfrak M}(x)$, as required.
\end{proof}	

As a consequence of Lemmas~\ref{Rpreserved} and  \ref{l:maxfiltr}, we obtain the usual `filtration lemma' for $\mathfrak M^\dag$ that can be proved by induction on $\chi$:

\begin{lemma}\label{lem:for} 
For any $\chi\in \sub(\varphi,\psi)$ and any $\p=(\ftype,\atype,\Type,M) \in W^\dag$, we have $\mathfrak{M}^\dag,\p\models \chi$ iff $\chi\in \ftype$.
\end{lemma}	
%
%
%
%
%
%

As a consequence of Lemma~\ref{lem:for} we obtain:

\begin{corollary}
$\mathfrak{M}^\dag,\p(\rr_\varphi) \models \varphi$ and $\mathfrak{M}^\dag,\p(\rr_\psi) \models \neg \psi$. 
\end{corollary}

Define a binary relation $\bis^\dag$ on $W^\dag$ by taking  
\[
(\ftype,\atype,\Type,M)\,\bis^\dag\, (\ftype',\atype',\Type',M')\quad\mbox{iff}\quad \atype=\atype',\ M=M'.
\]
%
%
%

\begin{lemma}
The relation $\bis^\dag$ is a $\varrho$-bisimulation on $\mathfrak{M}^\dag$ with $\p(\rr_\varphi) \bis^\dag \p(\rr_\psi)$.   
\end{lemma}
\begin{proof} 
As $\mathfrak M,\rr_\varphi \simr \mathfrak M,\rr_\psi$, we have $t^\varrho_{\mathfrak M}(\rr_\varphi)=t^\varrho_{\mathfrak M}(\rr_\psi)$.
So $\at^\varrho_{\mathfrak M}(\rr_\varphi)=\at^\varrho_{\mathfrak M}(\rr_\psi)$ and
$\cC{C(\rr_\varphi)}=\cC{C(\rr_\psi)}$, and hence $\p(\rr_\varphi) \bis^\dag \p(\rr_\psi)$.   

Condition {\bf (atom)} follows from Lemma~\ref{lem:for}. To prove {\bf (move)}, suppose $\p_1 \bis^\dag \p_1'$ and $\p_1 R^\dag \p_2$, 
for $\p_1=(\ftype_1,\atype_1,\Type_1,M_1)$, $\p_1'=(\ftype_1',\atype_1,\Type_1',M_1)$, and $\p_2=(\ftype_2,\atype_2,\Type_2,M_2)$.
We show that there is $\p_2'$ with $\p_2 \bis^\dag \p_2'$ and $\p_1' R^\dag \p_2'$, that is, 
there exist $\ftype_2'$ and $\Type_2'$ such that $\p_2'=(\ftype_2',\atype_2,\Type_2',M_2)\in W^\dag$ and $\p_1' R^\dag \p_2'$.
We proceed by case distinction.
\begin{description}
\item[\rm\emph{Case} $M_1\ne M_2$:] As $\p_1 R^\dag \p_2$, we have $M_1 \twoheadrightarrow M_2$.
As $\Type_1'\in M_1$, there is some $\Type_2'\in M_2$ with $\Type_1'\twoheadrightarrow\Type_2'$.
By \eqref{trans} and \eqref{tomax}, we may assume that $\Type_2'=\Type_D$ for some $\varrho$-maximal $D$, and so
$\Type_2'(\atype_2)\ne\emptyset$. Take any $\ftype_2'\in \Type_2'(\atype_2)$.
Then $\Type_2'$ and $\ftype_2'$ are as required, by \eqref{abstract}.
	
\item[\rm\emph{Case} $M_1 = M_2$, $M_1 \twoheadrightarrow M_1$:] 
As $\Type_1'\in M_1$, there is $\Type_2'\in M_1$ with $\Type_1'\twoheadrightarrow\Type_2'$.
By \eqref{trans} and \eqref{tomax}, we may assume that $\Type_2'=\Type_D$ for some $\varrho$-maximal $D$, and so
$\Type_2'(\atype_2)\ne\emptyset$. Take any $\ftype_2'\in \Type_2'(\atype_2)$.
Then $\p_2'=(\ftype_2',\atype_2,\Type_2',M_1)\in W^\dag$, by \eqref{abstract}.
As $\Type_1'\twoheadrightarrow\Type_2'$ implies $\ftype_1'\twoheadrightarrow\ftype_2'$, we have $\p_1' R^\dag \p_2'$ if 
$\p_1'=\p_2'$ or $\p_1'\ne\p_2'$.  
	
\item[\rm\emph{Case} $M_1 = M_2$, $M_1 \not\twoheadrightarrow M_2$, $\atype_1\ne\atype_2$:] 
If $\Type_1'(\atype_2)\ne\emptyset$, then take any $\ftype_2'\in \Type_2'(\atype_2)$. Then $\p_2'=(\ftype_2',\atype_2,\Type_1',M_1)$ is as required.
If $\Type_1'(\atype_2)=\emptyset$, then $\Type_1'\ne\Type_D$ for any $\varrho$-maximal $D$. Thus, by \eqref{irrall}, 
$\Type_1'(\atype)=\emptyset$ for every $\atype\in\iset{M_1}$. Also, by \eqref{tomax}, there is some $\varrho$-maximal $D$ with
$\Type_D\in M_1$ and $\Type_1'\twoheadrightarrow\Type_D$.
Take any $\ftype_2'\in \Type_D(\atype_2)$. By \eqref{abstract}, $\p_2'=(\ftype_2',\atype_2,\Type_D,M_1)$ is as required.
%
	
\item[\rm\emph{Case} $M_1 = M_2$, $M_1 \not\twoheadrightarrow M_2$, $\atype_1=\atype_2$:] We claim that $\atype_1\notin\iset{M_1}$.
Indeed, suppose $\atype_1\in\iset{M_1}$. As $\p_1 R^\dag\p_2$, $\p_1\ne\p_2$ follows. If $\Type_1=\Type_2$ held, then $\{\ftype_1\}=\Type_1(\atype_1)=\Type_2(\atype_1)=\{\ftype_2\}$ by \eqref{singleton}, and so $\ftype_1=\ftype_2$ would follow, contradicting $\p_1\ne\p_2$. So $\Type_1\ne\Type_2$, and thus $\p_1 R^\dag\p_2$ implies that 
$\Type_1(\atype)=\emptyset$ for all $\atype\in\iset{M_1}$. As $\ftype_1\in\Type_1(\atype_1)$, $\atype_1\notin\iset{M_1}$ follows, as required.


As $\atype_1\notin\iset{M_1}$, 
there exist $C$ with $M_1=\MM{C}$ and $\type\in\TC{C}$ such that $\type\cap\varrho=\atype_1$ and $\Diamond\type\not\subseteq\type$.
By Lemma~\ref{rho-clusters}, it follows that, for every $\varrho$-maximal $D\in\cC{C}$, either $(i)$ there are at least two $x\in D$ with $t^{\varrho}_{\mathfrak{M}}(x)=\type$ or $(ii)$ $xRx$ for the single $x\in D$ with $t^{\varrho}_{\mathfrak{M}}(x)=\type$. 

If $\Type_1'=\Type_{D}$ for a $\varrho$-maximal $D\in\cC{C}$, then in case $(ii)$ we have $\Type_1'(\atype_1)=
\{\ftype^{\varphi,\psi}_{\mathfrak M}(x)\}=\{\ftype_1'\}$ and $\ftype_1'\twoheadrightarrow\ftype_1'$.
Therefore, $\p_2'=(\ftype_1',\atype_1,\Type_1',M_1)$ is as required.
In case $(i)$, if $\ftype_1' \not\twoheadrightarrow \ftype_1'$, then there exists $\ftype\in\Type_1'(\atype_1)$ with $\ftype\ne\ftype_1'$,
and so $\p_2'=(\ftype,\atype_1,\Type_1',M_1)$ is as required.
	
If $\Type_1'\not=\Type_{D}$ for any $\varrho$-maximal $D\in\cC{C}$, then $\Type_1'(\atype)=\emptyset$ for every $\atype\in\iset{M_1}$,
by \eqref{irrall}. By \eqref{tomax},  there is a $\varrho$-maximal $D$ such that $\Type_D\in M_1$ and $\Type_1'\twoheadrightarrow \Type_D$. Take any $\ftype_2'\in \Type_D(\atype_1)$. By \eqref{abstract}, $\p_2'=(\ftype_2',\atype_1,\Type_D,M_1)$ is as required,
\end{description}
%
completing the proof of the lemma.
\end{proof}

The results obtained above yield the following:

\begin{theorem}
Any given implication $\varphi \to \psi$ does not have an interpolant in $\wKF$ iff there are models $\mathfrak M_\varphi$ and $\mathfrak M_\psi$ satisfying the criterion of Lemma~\ref{criterion} and having size triple-exponential in $|\varphi|$ and $|\psi|$.
\end{theorem}

Thus, to decide whether $\varphi \to \psi$ does not have an interpolant in $\wKF$, we can guess models $\mathfrak M_\varphi$ and $\mathfrak M_\psi$ of triple-exponential size in $|\varphi|$ and $|\psi|$ together with a binary relation $\simr$ between their points and check in polynomial time in the size of $\mathfrak M_\varphi$ and $\mathfrak M_\psi$ whether the conditions of Lemma~\ref{criterion} are met. 

\begin{theorem}
The IEP for $\wKF$ is decidable in \emph{\textsc{coN3ExpTime}}.
\end{theorem}

\begin{remark}\label{r:filtr}
We can use the above construction to check whether a formula $\psi$ is in $\wKF$ as follows. We clearly have
$\psi\notin\wKF$ iff $\varphi\to\psi$ has no interpolant in $\wKF$, for $\varphi=\psi\lor\neg\psi$. In this case, $\sigma=\varrho=\sig(\psi)$, and if $\mathfrak M$ is a descriptive $\sigma$-model, then
$\p(x)=\p(y)$ iff $\ftype^{\varphi,\psi}_{\mathfrak M}(x)=\ftype^{\varphi,\psi}_{\mathfrak M}(y)$ and $\ftype^{\varphi,\psi}_{\mathfrak M}\bigl(C(x)\bigr)=\ftype^{\varphi,\psi}_{\mathfrak M}\bigl(C(y)\bigr)$.
Also, our filtration becomes the $\sub(\psi)$-filtration given in \cite{kudinov2024filtrations}.
Note that this filtration gives a double-exponential bound on the size of the model satisfying $\neg\psi$, which is not optimal as 
the decision problem for $\wKF$ is $\PSpace$-complete~\cite{DBLP:journals/tocl/Shapirovsky22}.
\end{remark}

 
\section{Lower Bound}\label{lower}

\begin{theorem}\label{t:lowerbound}
The IEP for $\wKF$ is \coNExpTime-hard.
\end{theorem}

\begin{proof}
We show \NExpTime-hardness of interpolant non-existence by a reducion of the \emph{exponential torus tiling problem\/}. A \emph{tiling system} is a triple
$P=(T,H,V)$, where 
$T$ is a finite set of \emph{tile
	types} and $H,V \subseteq T \times T$ are the \emph{horizontal}
and \emph{vertical} matching conditions, respectively. 
An \emph{initial condition} for $P$ and $n>0$ takes the form $\bar{t} = (t_0,\dots,t_{n-1}) \in T^n$. A map 
$\tau \colon \{0,\dots,2^{n}-1\} \times \{0,\dots,2^{n}-1\} \to T$ is a
\emph{solution to $P$ and $\bar{t}$} if 
$\tau(i,0) = t_{i}$ for all $i < n$, and 
for all $i,j < 2^{n}$, the
following conditions hold (where $\,\opluss\,$ denotes addition modulo $2^n$):
\begin{itemize}
	\item[--] if $\tau(i,j) = t$ and $\tau(i \opluss 1,j) =
	t'$, then $(t,t') \in H$;
	
	\item[--] if $\tau(i,j) = t$ and $\tau(i,j \opluss 1) =
	t'$, then $(t,t') \in V$.
	
	%
\end{itemize}
It is well-known that the problem of deciding whether there is a
solution to given $P$ and $\bar{t}$ is \NExpTime-hard~\cite[Section 5.2.2]{DL-Textbook}.

Given a tiling system $P$ and an initial condition $\bar{t}$ of length $n>0$, we define formulas $\varphi$, $\psi$ of size polynomial in $|P|$ and $n$, such that, for $\sigma = \sig(\varphi) \cup \sig(\psi)$ and $\varrho = \sig(\varphi) \cap \sig(\psi)$,
\begin{align}
\nonumber
& \mbox{there is a solution for $P$ and $\bar{t}$\quad iff\quad}\\
\nonumber
& \hspace*{1cm}
\mbox{there exist $\sigma$-models $\mathfrak M_\varphi$ and $\mathfrak M_\psi$ based on frames for $\wKF$}\\
\label{iff}
& \hspace*{1cm}
 \mbox{with $\mathfrak{M}_\varphi,\rr_\varphi \simr \mathfrak{M}_\psi,\rr_\psi$, $\mathfrak M_\varphi, \rr_\varphi  \models \varphi$ and $\mathfrak M_\psi,\rr_\psi \models \neg\psi$.}
\end{align}
The shared signature $\varrho$ consists of
\begin{itemize}
\item[--]
a variable $\tvar$, for each tile type $t\in T$; 

\item[--]
variables $\bvar_0,\dots,\bvar_{2n-1}$, that serve as bits in the
binary representation of \emph{grid positions} $(i,j)$ with $i,j< 2^{n}$.
We will use $[\bvar=(i,j)]$ as a shorthand for the formula where 
$\bvar_0,\ldots,\bvar_{n-1}$ represent the horizontal coordinate $i$ and
$\bvar_n,\ldots,\bvar_{2n-1}$ the vertical coordinate $j$
(with $\bvar_0$ and $\bvar_n$ being the respective least significant bits); 
for instance, $[\bvar=(2,3)]$ stands for
$\neg \bvar_0\land \bvar_1\land\bigl(\bigwedge_{1< k <n}\neg{\bvar}_{k}\bigr)\land\bvar_{n}\land\bvar_{n+1}\land\bigwedge_{n+1< \ell <2n}\neg{\bvar}_{\ell}$; 

\item[--]
 a variable $\evar$ 
that will be used to establish connections between the two $\varrho$-bisimilar models that force the tiling 
matching conditions. 
\end{itemize}
The formula $\varphi$ is defined as
\[
\varphi= \evar\land(\Diamond\Diamond p \wedge \neg \Diamond p) \wedge \Box(\evar \rightarrow\Diamond p),
\]
to which we add the other symbols in $\varrho$ using tautologies.
Observe that $\varphi$ has the first formula of Example~\ref{ex:basic}~$(i)$  as a conjunct. Thus, 
if $\mathfrak M$ is based on a frame for $\wKF$ and 
$\mathfrak{M},x\models \varphi$, then
$x$ is irreflexive, $C(x)$ contains a point different from $x$, and $\evar$ is true everywhere in $C(x)$ and nowhere else $R$-accessible from $x$. 

Our formula $\psi$ takes form $\chi\to(\Box\Box q\to q)$ (cf.\ the second formula in Example~\ref{ex:basic}~$(i)$).
We next define $\chi$. 
To begin with, $\chi$ has conjuncts that
use variables $\avar_{0},\ldots,\avar_{2n-1}$ and variables $\lvar_{0},\ldots,\lvar_{2n}$ 
to generate a binary tree of depth $2n$ on nodes satisfying $\evar$ 
such that a counter implemented using $\avar_{0},\ldots,\avar_{2n-1}$ is realised at its leaves:
%
\begin{align}
\label{treefirst}
& \lvar_0\land\Box^+\!\!\!\bigwedge_{i<j\leq 2n}\!\!\!\neg(\lvar_i\land\lvar_j),\\
& \Box^{+}\bigl(\lvar_{i} \rightarrow \Diamond (\lvar_{i+1} \land \avar_{i}) \wedge \Diamond (\lvar_{i+1} \land \neg \avar_{i})\bigr),\quad\mbox{for $i<2n$},\\
\nonumber
& \Box\bigl(\lvar_{i+1} \land \avar_{i} \to \Box (\lvar_{j} \to \avar_{i})\bigr)\, \land\\
\label{treelast}
& \hspace*{1.7cm}
\Box\bigl(\lvar_{i+1} \land \neg \avar_{i} \to \Box (\lvar_{j} \to \neg \avar_{i}) \bigr),\quad\mbox{for $i<j\leq 2n$},\\
\label{alle}
& \Box^{+}(\lvar_{i} \to\evar),\quad\mbox{for $i\leq 2n$}.
\end{align}
Next, we express that any leaf making $[\avar=(i,j)]$ true has an $R$-successor making 
$\neg\evar\land[\bvar=(i,j)]$ and a unique tile-variable $\tvar$ true, by defining
\[
\grid:\quad  \bigwedge_{k<2n}(\avar_k\leftrightarrow\bvar_k),
\]
and then adding the following conjuncts to $\chi$:
%
\begin{align}
\label{gridgen}
&\Box \Bigl(\lvar_{2n} \to \Diamond \bigl(\neg\evar \land\grid\land\bigvee_{t\in T}\tvar\bigr)\Bigr),\\
\nonumber
& \Box\bigl(\lvar_{2n} \land \avar_{i} \to \Box (\neg\evar \to \avar_{i})\bigr)\,\land\\
\label{acont}
& \hspace*{2.5cm}
\Box\bigl(\lvar_{2n} \land \neg \avar_{i} \to \Box (\neg\evar \to \neg \avar_{i}) \bigr),\quad\mbox{for $i< 2n$},\\
\label{notiles}
& \Box\bigwedge_{t\ne t'\in T}\neg(\tvar\land\tvar'),\\
\label{tileunique}
& \Box\Bigl(\lvar_{2n}\to\bigwedge_{t\in T}\bigl(\Diamond(\neg\evar\land\grid\land\tvar)\to\Box(\neg\evar\land\grid\to\tvar)\bigr)\Bigr).
\end{align}
Recall that computing $\,\opluss\,$ for numbers $i<2^n$ in binary on $n$ bits is as follows: If $i=2^n-1$ then
flip all $1$-bits to $0$; otherwise, 
flip the first (when starting from the least significant bit) $0$-bit of $i$ to $1$, 
flip all $1$-bits of $i$ before the first $0$-bit to $0$, and leave all other bits of $i$ the same. So 
the following formulas, respectively, express that `$[\avar=(i,j)]$ and $[\bvar=(i\opluss 1,j)]$' and 
`$[\avar=(i,j)]$ and $[\bvar=(i,j\opluss 1)]$':
\begin{align*}
\gridx:\quad 
& \Bigl(\bigvee_{m<n}\bigl(\bvar_m\land\neg\avar_m\land
\bigwedge_{k<m}(\neg\bvar_k\land\avar_k)\land\bigwedge_{m<k<n}(\bvar_k\leftrightarrow\avar_k)\bigr)\,\lor\\
&\hspace*{4.5cm} 
\bigwedge_{m<n}(\neg\bvar_m\land\avar_m)\Bigr)\land
\bigwedge_{n\leq k<2n}(\bvar_k\leftrightarrow\avar_k),\\
\gridy:\quad 
& \Bigl(\bigvee_{n\leq m<2n}\bigl(\bvar_m\land\neg\avar_m\land
\bigwedge_{n\leq k<m}(\neg\bvar_k\land\avar_k)\land\bigwedge_{m<k<2n}(\bvar_k\leftrightarrow\avar_k)\bigr)\,\lor\\
&\hspace*{4.5cm} 
\bigwedge_{n\leq m<2n}(\neg\bvar_m\land\avar_m)\Bigr)\land
\bigwedge_{k<n}(\bvar_k\leftrightarrow\avar_k).
\end{align*}
%
%
We add the following conjuncts to $\chi$ to ensure the tiling matching conditions:
\begin{align}
\label{hmatch}
& \Box\Bigl(\lvar_{2n}\to\bigwedge_{t\in T}\bigl(\Diamond(\neg\evar\land\grid\land\tvar)\to\Box(\neg\evar\land\gridx\to\bigvee_{(t,t)'\in H}\tvar')\bigr)\Bigr),\\
\label{vmatch}
& \Box\Bigl(\lvar_{2n}\to\bigwedge_{t\in T}\bigl(\Diamond(\neg\evar\land\grid\land\tvar)\to\Box(\neg\evar\land\gridy\to\bigvee_{(t,t)'\in V}\tvar')\bigr)\Bigr).
\end{align}
Finally, we ensure that the initial condition $\bar{t}$ holds, that is $\tau(i,0)=t_{i}$ for $i<n$.  To this end, we add to $\chi$ the 
conjuncts 
\begin{equation}\label{init}
\Box\bigl(\lvar_{2n} \land [\avar=(i,0)] \to \Box(\neg\evar\land\grid\to \tvar_i)\bigr),\quad\mbox{for $i<n$.}
\end{equation}
It follows from the argument in 
Example~\ref{ex:basic}~$(i)$ that $(\varphi \to \psi) \in \wKF$. Below we show that \eqref{iff} holds. 

$(\Rightarrow)$ 
Suppose $\tau$ is a solution to $P$ and $\bar{t}$. 
We define $\sigma$-models $\mathfrak M_\varphi$ and $\mathfrak M_\psi$ as follows.
The underlying $\wKF$-frame of $\mathfrak M_\varphi$ consist of a two-element cluster $C$ having an irreflexive point $\rr_\varphi$ and 
a reflexive point $x$, and $C$ has $2^{2n}$ irreflexive and pairwise $R$-incomparable $R$-successors $w_{k,\ell}$, $k,\ell<2^n$.
The valuation in $\mathfrak M_\varphi$ is such that $p$ holds at $\rr_\varphi$, $\evar$ holds everywhere in $C$, and for $k,\ell<2^n$,
$\mathfrak M_\varphi,w_{k,\ell}\models[\bvar=(k,\ell)]\land\tvar$, where $t=\tau(k,\ell)$.
The underlying $\wKF$-frame of $\mathfrak M_\psi$ is the transitive closure of the following frame: First, take 
a full binary tree $(T,R_T)$ of depth $2n$, with an irreflexive root $\rr_\psi$,
all other nodes being reflexive, and having $2^{2n}$ leaves $e_{i,j}$, $i,j<2^n$.
Then, for each leaf $e_{i,j}\in T$, add $2^{2n}$ irreflexive $R$-successors $u^{k,\ell}_{i,j}$, $k,\ell<2^n$, such that 
$u^{k,\ell}_{i,j}$ and $u^{k',\ell'}_{i',j'}$ are $R$-incomparable whenever $(i,j,k,\ell)\ne(i',j',k',\ell')$. 
The valuation in $\mathfrak M_\psi$ is such that $q$ holds everywhere apart from $\rr_\psi$, $\evar$ holds everywhere in $T$,
$\lvar_0,\dots,\lvar_{2n}$ and $\avar_0,\dots,\avar_{2n-1}$ `mark' the nodes of the tree $(T,R_T)$ in such a way that, for $i,j< 2^n$,
$\mathfrak M_\psi,e_{i,j}\models\lvar_{2n}\land[\avar=(i,j)]$,
and for $i,j,k,\ell<2^n$,
$\mathfrak M_\psi,u^{k,\ell}_{i,j}\models[\avar=(i,j)]\land[\bvar=(k,\ell)]\land\tvar$, where $t=\tau(k,\ell)$.

It is 
straightforward to check that 
$\mathfrak M_\varphi, \rr_\varphi \models \varphi$ and $\mathfrak M_\psi,\rr_\psi\models \neg\psi$, and the relation
\[
\bis = (C\times T)\cup \big\{\bigl(w_{k,\ell},u_{i,j}^{k,\ell}\bigr)\mid i,j,k,\ell <2^n\bigr\}
\]
is a $\varrho$-bisimulation between $\mathfrak M_\varphi$ and $\mathfrak M_\psi$ with $\rr_\varphi\bis\rr_\psi$.\hfill $\dashv$
%

$(\Leftarrow)$ Suppose $\mathfrak{M}_\varphi,\rr_\varphi \simr \mathfrak{M}_\psi,\rr_\psi$ with $\mathfrak M_\varphi, \rr_\varphi \models \varphi$ and $\mathfrak M_\psi,\rr_\psi\models \neg\psi$,
for some $\sigma$-models $\mathfrak M_\varphi$ and $\mathfrak M_\psi$ based on frames for $\wKF$.
Thus, $\rr_\varphi$ is irreflexive, the cluster $C$ of $\rr_\varphi$ contains a point different from $\rr_\varphi$, and $\evar$ is true everywhere in $C$ and nowhere else in $\mathfrak M_\varphi$. As $\mathfrak M_\psi,\rr_\psi\models \neg\psi$ and  
by \eqref{treefirst}--\eqref{treelast}, $\rr_\psi$ is the root of a full binary tree of depth $2n$ having its leaves $e_{i,j}$, $i,j<2^n$, marked by 
$\lvar_{2n}$ and the corresponding formula $[\avar=(i,j)]$.
(Note that, as $\mathfrak M_\psi,\rr_\psi \models\Box\Box q\land\neg q$, none of $e_{i,j}$ is in $C(\rr_\psi)$.)
By \eqref{alle}, for all $i,j<2^n$, $\mathfrak M_\psi,e_{i,j}\models\evar$, and so we must have $\mathfrak M_\varphi, x_{i,j} \sim_\varrho \mathfrak M_\psi,e_{i,j}$ 
for some $x_{i,j}\in C$.
Thus, by \eqref{treelast} and \eqref{gridgen},
\begin{description}
\item[{\bf (t1)}]
for all $i,j<2^n$, $C$ has an $R$-successor $w$ such that $\mathfrak M_\varphi,w\models [\bvar=(i,j)]\land\tvar$ for some $t\in T$.
\end{description}
We claim that
\begin{description}
\item[{\bf (t2)}]
for all $i,j<2^n$, if $w$, $w'$ are $R$-successors of $C$ with $\mathfrak M_\varphi,w \models [\bvar=(i,j)]\land\tvar$ and
$\mathfrak M_\varphi,w'\models [\bvar=(i,j)]\land\tvar'$, then $t=t'$.
\end{description}
Indeed, by $\varrho$-bisimilarity, there exist $R$-successors $u$, $u'$ of $e_{i,j}$ such that 
$\mathfrak M_\psi,u\models \neg\evar\land[\bvar=(i,j)]\land\tvar$ and $\mathfrak M_\psi,u'\models \neg\evar\land[\bvar=(i,j)]\land\tvar'$.
By \eqref{acont}, $[\avar=(i,j)]$ is true at both $u$ and $u'$, and so $\grid$ is true at both $u$ and $u'$ as well.
Thus, $t=t'$ follows from \eqref{notiles} and \eqref{tileunique}.

Now we define a map $\tau$ by taking, for all $i,j<2^n$, 
$\tau(i,j)=t$ iff $C$ has an $R$-successor $w$ with $\mathfrak M_\varphi,w\models [\bvar=(i,j)]\land\tvar$.
By {\bf (t1)} and {\bf (t2)}, $\tau$ is well-defined. We claim that 
\begin{description}
\item[{\bf (t3)}]
for all $i,j<2^n$, if $\tau(i,j)=t$ and $\tau(i\opluss 1,j)=t'$ then $(t,t')\in H$, and

\item[{\bf (t4)}]
for all $i,j<2^n$, if $\tau(i,j)=t$ and $\tau(i,j\opluss 1)=t'$ then $(t,t')\in V$.
\end{description}
Indeed, for {\bf (t3)}, by the definition of $\tau$ and $\varrho$-bisimilarity, there are $R$-successors $u$, $u'$ of $e_{i,j}$ such that 
$$\mathfrak M_\psi, u \models \neg\evar\land[\bvar=(i,j)]\land\tvar, \quad \mathfrak M_\psi,u'\models \neg\evar\land[\bvar=(i\opluss 1,j)]\land\tvar'.$$
By \eqref{acont}, $[\avar=(i,j)]$ is true at both $u$ and $u'$. So $\grid$ is true at $u$ and $\gridx$ is true at $u'$, and so $(t,t')\in H$ follows from \eqref{hmatch}. Condition 
{\bf (t4)} can be shown similarly, using \eqref{vmatch} in place of \eqref{hmatch}.

It follows from {\bf (t3)}, {\bf (t4)} and \eqref{init} that $\tau$ is a solution to $P$ and  $\bar{t}$, completing the proof of \eqref{iff}.
Now the theorem follows by Lemma~\ref{criterion}.
\end{proof}


\section{Discussion}\label{discussion}

Our investigation of the interpolant existence problem for weak $\KF$ and the difference logic $\DL$ is part of a research programme that aims to understand Craig interpolants for logics not enjoying the CIP. It turns out that $\wKF$ shares with standard modal logics with nominals~\cite{DBLP:journals/tocl/ArtaleJMOW23}, decidable fragments of first-order modal logics~\cite{DBLP:conf/kr/KuruczWZ23}, and the guarded and two-variable fragment of first-order logic~\cite{DBLP:conf/lics/JungW21} that interpolant existence is still decidable but computationally harder than validity. In contrast, the difference logic $\DL$ shares with normal extensions of $\mathsf{K4.3}$~\cite{DBLP:journals/corr/abs-2312-05929} that interpolant existence has the same complexity
as validity. Linear temporal logic $\mathsf{LTL}$ is another example of a logic without the CIP, for which interpolant existence is decidable~\cite{DBLP:journals/corr/PlaceZ14}. For both $\wKF$ and $\mathsf{LTL}$, establishing tight complexity bounds remains an interesting open problem.
Further open problems include the following. As our decision procedures are non-constructive, it would be of interest to develop algorithms that compute interpolants whenever they exist. Also, is it possible to establish general decidability results for interpolant existence for families of extensions of $\wKF$? Another related question is to find out which additional logical connectives would repair the CIP for $\wKF$; see~\cite{tenCatediss,DBLP:journals/corr/abs-2310-08689} for elegant answers to such  questions for modal logics with nominals and other fragments of first-order logic.

The results of this paper are also relevant to the \emph{explicit definition existence problem} (EDEP) for $L\in \{\DL,\wKF\}$:  given formulas $\varphi, \psi$ and a signature $\varrho$, decide whether there is a $\varrho$-formula $\chi$ with $\varphi \rightarrow (\psi \leftrightarrow \chi)\in L$, called an \emph{explicit $\varrho$-definition of $\psi$ modulo $\varphi$ in $L$}. The EDEP reduces trivially to validity for logics enjoying the projective Beth definability property~\cite{MGabbay2005-MGAIAD}, which is not the case for $\DL$ and $\wKF$. In fact, one can prove in exactly the same way as for fragments of first-order modal logics~\cite{DBLP:conf/kr/KuruczWZ23} that, for $L\in \{\DL,\wKF\}$, the IEP and EDEP are polynomial-time reducible to each other.  
Thus, our results above also provide complexity bounds for the EDEP in $\DL$ and $\wKF$.


\medskip
\noindent
{\bf Acknowledgements.}
Thanks are due to the anonymous referees for their comments on the draft version of this paper.




\end{document}